\documentclass[11pt,reqno]{article}

\newcommand{\ifabs}[2]{#2}

\usepackage{amsmath, amssymb, amsfonts, bm, cite}
\renewcommand{\mathbf}[1]{\boldsymbol{#1}}

\usepackage{algorithm}

\usepackage{nicefrac}
\usepackage{pgf}

\usepackage{geometry}
\usepackage{fullpage}

\usepackage{pdfsync, graphicx}

\usepackage{comment}

\usepackage{amsthm}
\usepackage[colorlinks=true,citecolor=blue,bookmarksnumbered=true,hyperfootnotes=false]{hyperref}

\newtheorem{theorem}{Theorem}[section]

\newtheorem{lemma}[theorem]{Lemma}
\newtheorem{proposition}[theorem]{Proposition}

\newtheorem{definition}{Definition}[section]

\newtheorem{assumption}[theorem]{Assumption}

\newcommand{\floor}[1]{\ensuremath{\left\lfloor#1\right\rfloor}}

\newcommand{\E}[2][]{\ensuremath{\mathbb{E}_{#1}\insq{#2}}}

\newcommand{\insq}[1]{\left[#1\right]}

\newcommand{\concept}[1]{\emph{#1}}

\newcommand{\gibbs}{\mu}
\newcommand{\mgn}[2]{\gibbs^{#1}_{#2}}
\newcommand{\Rgibbs}[2]{\gibbs^{#1}_{#2}}
\newcommand{\Rmgn}[3]{\gibbs^{#1}_{#2,#3}}
\newcommand{\pr}{\mathrm{P}}
\newcommand{\clist}[1]{\mathcal{#1}}
\newcommand{\err}{\mathcal{E}}
\newcommand{\SAW}{\mathrm{SAW}}
\newcommand{\SABW}{\mathrm{SABW}}
\newcommand{\contr}{\delta}
\newcommand{\dist}{\mathrm{dist}}

\begin{document}
\title{Spatial Mixing of Coloring Random Graphs}

\author{ Yitong Yin\thanks{Supported by NSFC grants 61272081 and 61321491. }\\Nanjing University, China\\ \texttt{yinyt@nju.edu.cn}}


\date{}\maketitle
\begin{abstract}
We study the strong spatial mixing (decay of correlation) property of proper $q$-colorings of random graph $G(n, d/n)$ with a fixed~$d$.
The strong spatial mixing of coloring and related models have been extensively studied on graphs with bounded maximum degree.
However, for typical classes of graphs with bounded average degree, such as $G(n, d/n)$, an easy counterexample shows that colorings do not exhibit strong spatial mixing with high probability. 
Nevertheless, we show that for $q\ge\alpha d+\beta$ with $\alpha>2$ and sufficiently large $\beta=O(1)$,
with high probability proper $q$-colorings of random graph $G(n, d/n)$ exhibit strong spatial mixing \emph{with respect to an arbitrarily fixed vertex}. 
This is the first strong spatial mixing result for colorings of graphs with unbounded maximum degree. 
Our analysis of strong spatial mixing establishes  a block-wise correlation decay instead of the standard point-wise decay, which may be of interest by itself, especially for graphs with unbounded degree.




\end{abstract}

\section{Introduction}
\label{sec:introduction}
A proper $q$-coloring of a graph $G$ is an assignment of $q$ colors $\{1,2,\ldots, q\}$ to the vertices so that adjacent vertices receive different colors. 
Each coloring corresponds to a configuration in the $q$-state zero-temperature antiferromagnetic Potts model. The uniform probability space, known as the Gibbs measure, of proper $q$-colorings of the graph, receives extensive studies from both Theoretical Computer Science and Statistical Physics.

An important question concerned with the Gibbs measure is about the mixing rate of Glauber dynamics, usually formulated as: 
on graphs with {maximum degree} $d$, assuming $q\ge\alpha d+\beta$, the lower bounds for $\alpha$ and $\beta$ to guarantee rapidly mixing of the Glauber dynamics over proper $q$-colorings. (See~\cite{frieze2007survey} for a survey.)


Recently, much attention has been focused on the spatial mixing (correlation decay) aspect  of the Gibbs measure, 
which is concerned with the case where the site-to-boundary correlations in the Gibbs measure decay exponentially to zero with distance.
In Statistical Physics, spatial mixing implies the uniqueness of infinite-volume Gibbs measure. 
In Theoretical Computer Science, a stronger notion is considered: 
the \concept{strong spatial mixing} introduced in Weitz's thesis~\cite{Weitz04}. 
Here, the exponential decay of site-to-boundary correlations is required to hold even conditioning on an arbitrarily fixed boundary.
Strong spatial mixing is interesting to Computer Science because it may imply efficient approximation algorithms for counting and sampling.
This implication was fully understood for two-state spin systems. For multi-state spin systems such as coloring, this algorithmic implication of strong spatial mixing is only known for special classes of graphs, such as neighborhood-amenable (slow-growing) graphs~\cite{goldberg2005strong}.
Strong spatial mixing of proper $q$-coloring has been proved for classes of degree-bounded graphs, including regular trees~\cite{ge2011strong}, lattices graphs~\cite{goldberg2005strong}, and finally the general degree-bounded triangle-free graphs~\cite{gamarnik2012strong}, all with the same $\alpha>\alpha^*$ bound where $\alpha^*=1.763...$ is the unique solution to $x^x=\mathrm{e}$. 

All these temporal and spatial mixing results are established for graphs with bounded \emph{maximum degree}. It is then natural to ask what happens for classes of graphs with bounded \emph{average degree}. 
A natural model for the ``typical'' graphs with bounded average degree $d$ is the Erd\"os-R\'enyi random graph $G(n,d/n)$.
In this model, the Gibbs measure of proper $q$-colorings becomes more complicated because the maximum degree is unbounded and the decision of colorability is nontrivial.
Nevertheless, it was discovered in~\cite{col_DFFV06} that for $G(n,d/n)$ the rapid mixing of (block) Glauber dynamics over the proper $q$-colorings  can be guaranteed by a $q=O(\log\log n/\log\log\log n)$, much smaller than the maximum degree of $G(n,d/n)$. This upper bound on the number of colors was later reduced to a constant $q=\mathrm{poly}(d)$  in~\cite{efthymiou2007randomly} and independently in~\cite{mossel2008rapid,mossel2010gibbs}, and very recently to a linear $q\ge \alpha d+\beta$ with $\alpha=5.5$ in~\cite{efthymiou2013mcmc}.

On the spatial mixing side, the strong spatial mixing of the models which are simpler than coloring has been studied on random graph $G(n,d/n)$, or other classes of graphs with bounded average degree.
Recently in~\cite{sinclair2013spatial}, such average-degree based strong spatial mixing is established for the independent sets of graphs with bounded \concept{connective constant}. Since $G(n,d/n)$ has connective constant $\approx d$ with high probability, this result is naturally translated to $G(n,d/n)$.

It is then an important open question to ask about the conditions for the spatial mixing of colorings of graphs with bounded average degree.
The following simple example shows that this can be very hard to achieve: Consider a  long path of $\ell$ vertices, each adjacent to $q-2$ isolated vertices, where $q$ is the number of colors. When the path is sufficiently long, the connective constant of this graph can be arbitrarily close to 1. However, colors of those isolated vertices can be properly fixed 
to make the remaining path effectively a 2-coloring instance, which certainly has long-range correlation, refuting the existence of strong spatial mixing.

More devastatingly, it is easy to see that for any constant $q$, with high probability the random graph $G(n,d/n)$ contains a path of length $\ell=\Theta(\log n)$ in which every vertex has degree $q-2$. As in the above example, even in a weaker sense of site-to-site correlation which was considered in~\cite{goldberg2005strong}, this forbids the strong spatial mixing up to a distance $\Theta(\log n)$. Meanwhile, it is well known that the diameter of $G(n,d/n)$ is $O(\log n)$ with high probability. 
So the strong spatial mixing of colorings of random graph $G(n,d/n)$ cannot hold except for a narrow range of distances in $\Theta(\log n)$.

In this case, inspired by the studies of spatial mixing in rooted trees, where only the decay of correlation to the root is considered,
we propose to study the strong spatial mixing \emph{with respect to a fixed vertex}, instead of all vertices.



\begin{assumption}\label{assumption}
We make following assumptions:
\begin{itemize}
\item $d> 1$ is fixed, and $q\ge\alpha d+\beta$ for $\alpha>2$ and sufficiently large $\beta=O(1)$ ($\beta\ge 23$ is fine);
\item $v\in V$ is arbitrarily fixed and $G=(V,E)$ is a random graph drawn from $G(n,d/n)$, where $n$ is sufficiently large. 
\end{itemize}
\end{assumption}

Note that vertex $v$ is fixed independently of the sampling of random graph.
With these assumptions we prove the following theorem. 
\begin{theorem}\label{theorem-ssm-random-graph}
Let  $q,v$ and $G$ satisfy Assumption~\ref{assumption},
and $t(n)=\omega(1)$ an arbitrary super-constant function.
With high probability, $G$ is $q$-colorable and the following holds: 
for any region $R\subset V$ containing $v$, whose vertex boundary is $\partial R$, for any feasible colorings $\sigma,\tau\in[q]^{\partial R}$ partially specified on $\partial R$ which differ only at vertices that are at least $t(n)$ distance away from $v$ in $G$,  
for some constants $C_1,C_2>0$ depending only on $d$ and $q$, it holds that
\[
|\Pr[c(v)=x\mid \sigma]-\Pr[c(v)=x\mid \tau]|\le C_1\exp(-C_2\cdot\dist(v,\Delta)),
\]
for a uniform random proper $q$-coloring $c$ of $G$ and any $x\in[q]$,
where $\Delta\subset\partial R$ is the vertex set on which $\sigma$ and $\tau$ differ, and $\dist(v,\Delta)$ denotes the shortest distance in $G$ between $v$ and any vertex in $\Delta$.
\end{theorem}

This is the first strong spatial mixing result for colorings of graphs with unbounded maximum degree.
Our technique is developed upon the error function method introduced in~\cite{gamarnik2012strong}, which uses a cleverly designed error function to measure the discrepancy of marginal distributions, and the strong spatial mixing is implied by an exponential decay of errors measured by this function.

In all existing techniques for strong spatial mixing of colorings, when the degree of a vertex is unbounded, a multiplicative factor of $\infty$ is contributed to the decay of correlation, which unavoidably ruins the decay.
However, in the real case for colorings of graphs with unbounded degree, a large-degree vertex may at most locally ``freeze'' the coloring, rather than nullify the existing decay of correlation. This limitation on the effect of large-degree vertex has not been addressed by any existing techniques for spatial mixing. 

We address this issue by considering a block-wise correlation decay, so that  within a block the coloring might be ``frozen'', but between blocks, 
the decay of correlation is as in that between vertices in the degree-bounded case.
This analysis of block-wise correlation decay can be seen as a spatial analog to the block dynamics over colorings of random graphs, and is the first time that such an idea is used in the analysis of spatial mixing.





\paragraph{Related work}
As one of the most important random CSP, the decision problem of coloring sparse random graphs has been extensively studied, e.g. in~\cite{achlioptas2005two, coja2013chasing}. 
Monte Carlo algorithms for sampling random coloring in sparse random graphs were studied in~\cite{dyer2010randomly, efthymiou2007randomly, mossel2008rapid, mossel2010gibbs, efthymiou2013mcmc},  and in~\cite{efthymiou2012simple}, a non-Monte-Carlo algorithm was given for the same problem which
uses less colors but has worse error dependency than the Monte-Carlo algorithms.
In~\cite{GK07,lu2013improved} the correlation decay on computation tree for coloring was studied which implies FPTAS for counting coloring.

\section{Preliminaries}
\label{sec:preliminaries}
\paragraph{Graph coloring.}
Let $G=(V,E)$ be an undirected graph. For each vertex $v\in V$, let $d_G(v)$ denote the degree of $v$. For any $u,v\in V$, let $\dist_G(u,v)$ denote the distance between $u$ and $v$ in $G$; and for any vertex sets $S, T\subseteq V$, let $\dist_G(u, S)=\min_{v\in S}\dist_G(u,v)$ and $\dist_G(S,T)=\min_{u\in S, v\in T}\dist_G(u,v)$. The subscripts can be omitted if graph $G$ is assumed in context. For any vertex set $S\subset V$, we use $\partial S=\{v\not\in S\mid uv\in E,u\in S\}$ to denote the \concept{vertex boundary} of $S$, and use $\delta S=\{uv\in E\mid u\in S,v\not\in S\}$ to denote the \concept{edge boundary} of $S$.

We consider the \emph{list-coloring problem}, which is a generalization of $q$-coloring problem. Let $q>0$ be a finite integer, a pair $(G,\clist{L})$ is called a \concept{list-coloring instance} if $G=(V,E)$ is an undirected graph, and $\clist{L}=(L(v): v\in V)$ is a sequence of lists where for each vertex $v\in V$, $L(v)\subseteq [q]$ is a list of colors from $[q]=\{1,2,\ldots,q\}$ associated with vertex $v$. A $\sigma\in[q]^V$ is a \concept{proper coloring} of $(G,\clist{L})$ if $\sigma(v)\in L(v)$ for every vertex $v\in V$ and no two adjacent vertices in $G$ are assigned with the same color by $\sigma$. A list-coloring instance $(G,\clist{L})$ is said to be \concept{feasible} or \concept{colorable} if there exists a proper coloring of $(G,\clist{L})$. A coloring can also be partially specified on a subset of vertices in $G$. For $S\subseteq V$, let $L(S)=\{\sigma\in[q]^S\mid \forall v\in V,\sigma(v)\in L(v)\}$ denote the set of all possible colorings (not necessarily proper) of the vertices in $S$. A coloring $\sigma\in L(S)$ partially specified on a subset $S\subseteq V$ of vertices is said to be \concept{feasible} if there is a proper coloring $\tau$ of $(G,\clist{L})$ such that $\sigma$ and $\tau$ are consistent over set $S$. A coloring $\sigma\in L(S)$ partially specified on a subset $S\subseteq V$ of vertices is said to be \concept{proper} or \concept{locally feasible} if $\sigma$ is a proper coloring of $(G[S], \clist{L}_S)$ where $G[S]$ is the subgraph of $G$ induced by $S$ and $\clist{L}_S=(L(v):v\in S)$ denotes the sequence $\clist{L}$ of lists restricted on set $S$ of vertices. For any $S\subseteq V$, we use $L^*(S)$ to denote the set of proper colorings of $S$.

When $L(v)=[q]$ for all vertices $v\in V$, a list-coloring instance $(G,\clist{L})$ becomes an instance for $q$-coloring, 
which we denote as $(G,[q])$.



\paragraph{Self-avoiding walk (SAW) tree.}
Given a graph $G(V, E)$ and a vertex $v\in V$, a tree $T$ rooted by $v$ can be naturally constructed from all self-avoiding walks starting from $v$ so that each walk corresponds to a vertex in $T$, and each walk $p$ is the parent of walks $(p,u)$ where $u\in V$ is a vertex. We use $T_\SAW(G,v)=T$ to denote this tree constructed as above, and call it a \concept{self-avoiding walk tree (SAW)} of graph $G$.

\paragraph{Gibbs measure and strong spatial mixing.}
A feasible list-coloring instance $(G,\clist{L})$ gives rise to a natural probability distribution $\mu=\mu_{G,\clist{L}}$, which is the uniform distribution over all proper list-colorings. This distribution $\mu$ is also called the \concept{Gibbs measure} of list-colorings. We also a notation of  $\pr_{G,\clist{L}}(\text{event}(c))=\Pr[\text{event}(c)]$ to evaluate probability of an event defined on a uniform random proper coloring $c$ of $(G,\clist{L})$.
Let $B\subset V$ and $\Lambda\subset V$. For any feasible coloring $\sigma\in L(\Lambda)$ partially specified on vertex set $\Lambda$, we use $\mu_B^\sigma=\mu_{G,\clist{L},B}^\sigma$ to denote the marginal distribution over colorings of vertices in $B$ conditioning on that the coloring of vertices in $\Lambda$ is as specified by $\sigma$. And when $B=\{v\}$, we write $\mu_v^\sigma=\mu_{G,\clist{L},v}^\sigma=\mu_{G,\clist{L},\{v\}}^\sigma$. The list-coloring instance $(G,\clist{L})$ in the subscripts can be omitted if it is assumed in context.
Formally, for a uniformly random proper coloring $c$ of $(G,\clist{L})$, we have
\begin{align*}
\forall x\in L(v),
&\quad \mu_v^\sigma(x)=\pr_{G,\clist{L}}(c(v)=x\mid \sigma),\\ 
\forall \pi\in L(B),
&\quad \mu_B^\sigma(\pi)=\pr_{G,\clist{L}}(c(B)=\pi\mid \sigma). 
\end{align*}



The notion strong spatial mixing is introduced in~\cite{Weitz04,Weitz06} for independent sets and extended to colorings in~\cite{goldberg2005strong,gamarnik2012strong}.
\begin{definition}[Strong Spatial Mixing]
The Gibbs measure on proper $q$-colorings of a family $\mathcal{G}$ of finite graphs exhibits \concept{strong spatial mixing (SSM)} if there exist constants $C_1,C_2>0$ such that for any graph $G(V,E)\in\mathcal{G}$, any $v\in V, \Lambda\subseteq V$, and any two feasible $q$-colorings $\sigma,\tau\in[q]^\Lambda$, we have
\[
  \|\gibbs^{\sigma}_v-\gibbs^{\tau}_v\|_{\mathrm{TV}}\le C_1\exp(-C_2\mathrm{dist}(v,\Delta)),
\]
where $\Delta\subseteq\Lambda$ is the subset on which $\sigma$ and $\tau$ differ, and $\|\cdot\|_{\mathrm{TV}}$ is the total variation distance. 
\end{definition}
When the exponential bound relies on $\mathrm{dist}(v,\Lambda)$ instead of $\mathrm{dist}(v,\Delta)$, the definition becomes \concept{weak spatial mixing (WSM)}. The difference is SSM requires the exponential correlation decay continues to hold even conditioning on the coloring of a subset $\Lambda\setminus\Delta$ of vertices being arbitrarily (but feasibly) specified.


\paragraph{Random graph model}
The Erd\"os-R\'enyi random graph $G(n,p)$ is the graph with $n$ vertices $V$ and random edges $E$ where for each pair $\{u,v\}$, the edge $uv$ is chosen independently with probability $p$. We consider $G(n,d/n)$ with fixed $d>1$.


We say an event occurs \emph{with high probability (w.h.p.)} if the probability of the event is $1-o(1)$.

\newcommand{\SSSSMDefinition}{
\begin{definition}[Single-Site Strong Spatial Mixing]
The Gibbs measure on proper $q$-colorings of $G(V,E)$ exhibits signle-site strong spatial mixing at rate $\delta(\cdot)$ if for any $v\in V, \Lambda\subseteq V$ and any two feasible $\sigma_\Lambda,\tau_\Lambda\in[q]^\Lambda$ disagreeing on one vertex $u\in\Lambda$,
  \[
\|\mgn{\sigma_\Lambda}{v}-\mgn{\tau_\Lambda}{v}\|_{\mathrm{TV}}\le \delta(\mathrm{dist}(v,u)).
  \]
\end{definition}

\todo{The following is an alternative definition of SSSSM}

A vertex set $R\subseteq V$ is called a region, and $\partial R=\{\{u,v\}\in E\mid u\in R,v\not\in R\}$ denotes its edge boundary. An edge coloring $\tau\in[q]^{\partial R}$ of the boundary $\partial R$ is called an edge-boundary condition for region $R$. A configuaration $\sigma\in[q]^R$ is a proper $q$-coloring of region $R$ with boundary condition $\tau\in[q]^{\partial R}$ if $\sigma$ is a proper $q$-coloring of the induced subgraph $G[R]$ and $\sigma(u)\neq\tau(e)$ for every boundary edge $e=\{u,v\}\in\partial R$ with $u\in R$. 
Given a region $R$ and an $R$-feasible edge-boundary condition $\tau\in[q]^{\partial R}$, we use $\Rgibbs{\tau}{R}$ to denote the uniform distribution over all proper $q$-coloring of region $R$ with boundary condition $\tau$, and use $\Rmgn{\tau}{R}{v}$ to denote the correpsonding marginal distribution at vertex $v\in R$.

\begin{definition}[Single-Site Strong Spatial Mixing]
The Gibbs measure on proper $q$-colorings of $G(V,E)$ exhibits single-site strong spatial mixing at rate $\delta(\cdot)$ if for any region $R\subseteq V$, vertex $v\in R$, and any two $R$-feasible edge-boundary conditions $\sigma,\tau\in[q]^{\partial R}$ which differ at one boundary edge $e\in\partial R$,
\[
\|\Rmgn{\sigma}{R}{v}-\Rmgn{\tau}{R}{v}\|_{\mathrm{TV}}\le \delta(\mathrm{dist}(v,e)).
  \]
\end{definition}
}





\section{Correlation decay along self-avoiding walks}
\label{sec:SAW-tree}
In this section, we analyze the propagation of errors between marginal distributions measured by a special norm introduced in~\cite{gamarnik2012strong} in general degree-unbounded graphs.
Throughout this section, we assume $(G,\clist{L})$ to be a list-coloring instance with $G=(V,E)$ and $\clist{L}=(L(v): v\in V)$ where each $L(v)\subseteq[q]$. 

The following error function is introduced in~\cite{gamarnik2012strong}.

\begin{definition}[error function]
Let $\mu_1:\Omega\rightarrow [0,1]$ and $\mu_2:\Omega\rightarrow[0,1]$ be two probability measures over the same sample space $\Omega$. We define
\begin{align*}
\err(\mu_1,\mu_2)
&=\max_{x,y\in \Omega}\left(\log\left(\frac{\mu_1(x)}{\mu_2(x)}\right)-\log\left(\frac{\mu_1(y)}{\mu_2(y)}\right)\right),
\end{align*}
with the convention that $0/0=1$ and $\infty-\infty=0$.
\end{definition}

We assume $(G,\clist{L})$ to be feasible so that for vertex set $B\subset V$ and feasible colorings $\sigma,\tau\in L(\Lambda)$ of vertex set $\Lambda\subset V$, the marginal probabilities $\mu_B^\sigma$ and $\mu_B^\tau$ are well-defined.
The strong spatial mixing is proved by establishing a propagation of errors $\err(\mu_B^\sigma,\mu_B^\tau)$. Note that unlike in bounded-degree graphs, in general the value of $\err(\mu_B^\sigma,\mu_B^\tau)$ can be infinite, which occurs when the possibility of a particular coloring of $B$ is changed by  conditioning on $\sigma$ and $\tau$. This is avoided when a vertex cut with certain ``permissive'' property separating $B$ from the boundary. \ifabs{The following proposition is proved in Appendix~\ref{appendix-SAW}.}{}

\begin{proposition}\label{proposition-permissive-boundary}
If there is a $S\subset V\setminus(B\cup\Lambda)$ such that $|L(v)|>d(v)+1$ for every $v\in S$ and removing $S$ disconnects $B$ and $\Lambda$, 
then $\err(\mu_{B}^\sigma,\mu_{B}^\tau)$ is finite for any feasible colorings $\sigma,\tau\in L(\Lambda)$. 
\end{proposition}
\newcommand{\ProofPermissiveBoundary}{
It is sufficient to show that $\pr_{G,\clist{L}}(c(B)=\pi\mid\sigma)>0$ whenever $\pr_{G,\clist{L}}(c(B)=\pi\mid \tau)>0$.
Suppose that removing $B$ separates the graph into subgraphs $G_1$ and $G_2$ where $G_1$ contains $B$ and $G_2$ contains $\Lambda$.
For any proper coloring $\pi'$ of $G_1$ and any proper coloring $\sigma'$ of $G_2$, we must have $\pr_{G,\clist{L}}(c(G_1)=\pi'\wedge c(G_2)=\sigma')>0$, because $|L(v)|>d(v)+1$ for every $v\in S$, and hence it is always possible to coloring $S$ in a greedy fashion to complete a proper coloring $\pi'$ of $G_1$ along with a proper coloring $\sigma'$ of $G_2$ to a proper coloring of the entire graph $G$. Note that this implies the lemma because now a coloring $\pi$ of $B$ is possible if and only if it can be completed to a proper coloring of $G_1$, a property  independent of $\sigma$ and $\tau$.
}
\ifabs{
}{
\begin{proof}
\ProofPermissiveBoundary
\end{proof}
}

This motivates the following definition of permissive vertex and vertex set.
\begin{definition}
Given a list-coloring instance $(G,\clist{L})$, a vertex $v$ is said to be \concept{permissive} in $(G,\clist{L})$ if for all neighbors $u$ of $v$ and $u=v$, it holds that $|L(u)|>d(u)+1$.
A set $S$ of vertices is said to be permissive if all vertices in $S$ are permissive.
\end{definition}


Let $T=T_{\SAW}(G,v)$ be the self-avoiding walk tree of graph $G$ expanded from vertex $v$. 
Recall that every vertex $u$ in $T$ can be naturally identified (many-to-one) with the vertex in $G$ at which the corresponding self-avoiding walk ends (which we also denote by the same letter $u$).


\begin{definition}\label{definition-SAW-decay}
Given a list-coloring instance $(G,\clist{L})$, let $v\in V$, $T=T_{\SAW}(G,v)$, and $S$  a set of vertices in $T$. 
Suppose that the root $v$ has $m$ children $v_1,v_2,\ldots,v_m$ in $T$ and for $i=1,2\ldots,m$, let $T_i$ denote the subtree rooted by $v_i$. The quantity $\err_{T,\clist{L},S}$ is recursively defined as follows
\begin{align*}
\err_{T,\clist{L},S}
=
\begin{cases}
\displaystyle\sum\limits_{i=1}^m \contr\left(v_i\right) \cdot\err_{T_i,\clist{L},S} & \mbox{if }v\not\in S,\\
3q & \mbox{if }v\in S,
\end{cases}
\end{align*}
where $\contr(u)$ is a piecewise function defined as that for any vertex $u$ in $T$,
\begin{align*}
\contr(u)
=
\begin{cases}
\frac{1}{|L(u)|)-d_G(u)-1} & \mbox{if } |L(u)|>d_G(u)+1,\\
1 & \mbox{otherwise},
\end{cases}
\end{align*}
where $d_G(v)$ is the degree in the original graph $G$ instead of the degree in SAW-tree $T$.

In particular, when $(G,\clist{L})$ is a $q$-coloring instance $(G,[q])$, we denote this quantity as $\err_{T,[q],S}$.
\end{definition}



To state the main theorem of this section, we need one more definition.
\begin{definition}\label{definition-cutset}
Let $G=(V,E)$, $v\in V$, $\Delta\subset V$, and $T=T_{\SAW}(G,v)$. A set $S$ of vertices in $T$ is a \concept{cutset} in $T$ for $v$ and $\Delta$ if: (1) no vertex in $S$ is identified to $v$ or any vertex $u$ with $\dist(u,\Delta)<2$ by $T_\SAW(G,v)$; and (2) any self-avoiding walk from $v$ to a vertex in $\Delta$ must intersect $S$ in $T$. A cutset $S$ in $T$ for $v$ and $\Delta$ is said to be permissive in $(G,\clist{L})$ if every vertex in $S$ is identified with a permissive vertex in $(G,\clist{L})$ by $T_\SAW(G,v)$.
\end{definition}

The following theorem is the main theorem of this section, which bounds the error function $\err(\mu_{v}^\sigma,\mu_{v}^\tau)$ by the $\err_{T,\clist{L},S}$ defined in Definition~\ref{definition-SAW-decay} when there is a good cutset in the SAW tree.

\begin{theorem}\label{thm-SAW-decay}
Let $(G,\clist{L})$ be a feasible list-coloring instance where $G=(V,E)$ and $\clist{L}=(L(v)\subseteq[q]: v\in V)$. 
Let $v\in V$, $\Lambda\subset V$ and $\Delta\subseteq\Lambda$ be arbitrary, and $T=T_{\SAW}(G,v)$.
If there is a permissive cutset $S$ in $T$ for $v$ and $\Delta$, 
then for any feasible colorings $\sigma,\tau\in L(\Lambda)$ 
which differ only on $\Delta$, it holds that
\[
\err(\mu_{v}^\sigma,\mu_{v}^\tau)\le \err_{T,\clist{L},S}.
\]
\end{theorem}

This theorem is implied by the following weak spatial mixing version of the theorem.
\begin{lemma}\label{lemma-SAW-decay}
Let $(G,\clist{L})$ be a feasible list-coloring instance where $G=(V,E)$ and $\clist{L}=(L(v)\subseteq[q]: v\in V)$. 
Let $v\in V$ and $\Delta\subseteq\Lambda$ be arbitrary, and $T=T_{\SAW}(G,v)$.
If there is a permissive cutset $S$ in $T$ for $v$ and $\Delta$, 
then for any feasible colorings $\sigma,\tau\in L(\Delta)$, it holds that
\[
\err(\mu_{v}^\sigma,\mu_{v}^\tau)\le \err_{T,\clist{L},S}.
\]
\end{lemma}

\newcommand{\ProofThmSAWDecay}{
The two feasible colorings $\sigma,\tau\in L(\Lambda)$ can be expressed as $\sigma=(\sigma',\eta)$ and $\tau=(\tau',\eta)$ such that $\sigma'$ and $\tau'$ are two feasible colorings of vertices in $\Delta$ and $\eta$ is a feasible coloring of vertices in $\Gamma=\Lambda\setminus\Delta$.
Let $(G_\Gamma,\clist{L}_\eta)$ be such a list-coloring instance where $G_\Gamma$ is obtained from $G$ by deleting all vertices in $\Gamma$ and incident edges, and $\clist{L}_\eta$ is a color list for vertices in $G_\Gamma$ obtained from $\clist{L}$ by deleting color $\eta(u)$ from the lists $L(w)$ for all neighbors $w$ of any $u\in \Gamma$. 
Clearly, $(G_\Gamma,\clist{L}_\eta)$ is the instance obtained from $(G,\clist{L})$ by conditioning on that $\Gamma$ is colored as $\eta$, thus $(G_\Gamma,\clist{L}_\eta)$ must be feasible since $\eta$ is feasible.
Let $T'=T_\SAW(G_\Gamma,v)$. Obviously $T'$ is a subtree of $T=T_\SAW(G,v)$. 
Let $S'\subseteq S$ be obtained from permissive cutset $S$ in $T$ for $v$ and $\Delta$ by excluding those vertices which are identified with a vertex in $\Gamma$ by $T_{\SAW}(G,v)$. 
It is easy to see that $S'$ is a cutset in $T'$ for $v$ and $\Delta$, and $S'$ is also permissive in $(G_\Gamma,\clist{L}_\eta)$ because the operation applied by $(G_\Gamma,\clist{L}_\eta)$ on the original instance $(G,\clist{L})$ never decreases the gap $|L(u)|-d(u)$.
Thus, by Lemma~\ref{lemma-SAW-decay}, we have 
\[
\err(\mu_1',\mu_2')\le \err_{T',\clist{L}_\eta,S'},
\]
where $\mu_1'=\mu_{G_\Gamma,\clist{L}_\eta,v}^{\sigma'}$ and $\mu_2'=\mu_{G_\Gamma,\clist{L}_\eta,v}^{\tau'}$ are the marginal distributions at $v$ in the new instance $(G_\Gamma,\clist{L}_\eta)$. 
It is easy to see that 
\begin{align*}
\pr_{G,\clist{L}}(c(v)=x\mid \sigma)
&=\pr_{G,\clist{L}}(c(v)=x\mid \sigma',\eta)
=\pr_{G_\Gamma,\clist{L}_\eta}(c(v)=x\mid \sigma')=\mu_1'(x),\\
\pr_{G,\clist{L}}(c(v)=x\mid \tau)
&=\pr_{G,\clist{L}}(c(v)=x\mid \tau',\eta)
=\pr_{G_\Gamma,\clist{L}_\eta}(c(v)=x\mid \tau')=\mu_2'(x),
\end{align*}
thus $\mu_v^\sigma=\mu_1'$ and $\mu_v^\tau=\mu_2'$ where $\mu_v^\sigma=\mu_{G,\clist{L},v}^\sigma$ and $\mu_v^\tau=\mu_{G,\clist{L},v}^\tau$ are marginal probabilities defined in the original instance $(G,\clist{L})$. 
Therefore, we have $\err(\mu_v^\sigma,\mu_v^\tau)=\err(\mu_1',\mu_2')\le \err_{T',\clist{L}_\eta,S'}$. It remains to show that $\err_{T',\clist{L}_\eta,S'}\le \err_{T,\clist{L},S}$ where $T=T_\SAW(G,v)$, which is quite easy to see, because every self-avoiding walk in $G_\Gamma$ ended in $S'$ must be a self-avoiding walk in $G$ ended in $S$  and also the operation applied by $(G_\Gamma,\clist{L}_\eta)$ on the original instance $(G,\clist{L})$ never decreases the gap $|L(u)|-d(u)$ thus never increases the value of $\contr(u)$ for any vertex $u$ in the SAW-tree.
}
\ifabs{
The implication from Lemma~\ref{lemma-SAW-decay} to Theorem~\ref{thm-SAW-decay} is quite standard, whose proof is in Appendix~\ref{appendix-SAW}. It now remains to prove Lemma~\ref{lemma-SAW-decay}.
}
{
\begin{proof}[Proof of Theorem~\ref{thm-SAW-decay} by Lemma~\ref{lemma-SAW-decay}]
\ProofThmSAWDecay
\end{proof}
}

\subsection{The block-wise correlation decay}
Now our task is to prove Lemma~\ref{lemma-SAW-decay}. This is done by establishing the decay of $\err(\mu_{B}^\sigma,\mu_{B}^\tau)$ along walks among blocks $B$ with the following good property.
\begin{definition}
Given a list-coloring instance $(G,\clist{L})$, a vertex set $B\subseteq V$ is a \concept{permissive block around $v$} in $(G,\clist{L})$ if $v\in B$ and $|L(u)|>d_G(u)+1$ for every vertex $u$ in the vertex boundary $\partial B$.
\end{definition}

For permissive blocks $B$, a coloring of $B$ is globally feasible if and only if it is locally feasible (i.e.~proper on $B$).
\begin{lemma}\label{lemma-local-feasible}
Let $\Delta\subset V$ and $B\subset V$ a permissive block such that $\mathrm{dist}(B,\Delta)\ge 2$. Then for any feasible coloring $\sigma\in L(\Delta)$, for any coloring $\pi\in L(B)$, it holds that $\mu_B^\sigma(\pi)>0$ if and only if $\pi$ is proper on $B$. 
\end{lemma}
\begin{proof}
Let $S=\partial B$. Note that with $\mathrm{dist}(B,\Delta)\ge 2$ and $S$ must be a vertex cut separating $B$ and $\Delta$. Then the lemma can be proved by the same argument as in the proof of Proposition~\ref{proposition-permissive-boundary}.
\end{proof}


\noindent\textbf{Notations.} We now define some notations which are used throughout this section. 
Let $B\subset V$ be a permissive block in a feasible list-coloring instance $(G,\clist{L})$.
Let $\delta B=\{uw\in E\mid u\in B\mbox{ and }w\not\in B\}$ be the edge boundary of $B$. 
We enumerate these boundary edges as $\delta B=\{e_1,e_2,\ldots,e_m\}$.
For $i=1,2,\ldots,m$, we assume $e_i=u_iv_i$ where $u_i\in B$ and $v_i\not\in B$. Note that in this notation more than one $u_i$ or $v_i$ may refer to the same vertex in $G$. Let $G_B=G[V\setminus B]$ be the subgraph of $G$ induced by vertex set $V\setminus B$.  
For a coloring $\pi\in L(B)$ and $1\le i\le m$, we denote $\pi_i=\pi(u_i)$. 
For $1\le i\le m$ and $\pi,\rho\in L(B)$, let $\clist{L}_{i,j,\pi,\rho}=(L'(v): v\in V\setminus B)$ be 
obtained from $\clist{L}$ by removing the color $\pi_k$ from the list $L(v_k)$ for all $k<i$ and removing the color $\rho_k$ from the list $L(v_k)$ for all $k>i$ (if any of these lists do not contain the respective color then no change is made to them).

\ifabs{
With this notation, the following lemma generalizes a recursion introduced in~\cite{gamarnik2012strong} for bounded-degree graphs to general graphs. The proof is in Appendix~\ref{appendix-SAW}.
}
{With this notation, the following lemma generalizes a recursion introduced in~\cite{gamarnik2012strong} for bounded-degree graphs to general graphs by using permissive blocks.
}
\begin{lemma}\label{lemma-partial-pinning-invariant}
Let $(G,\clist{L})$ be a feasible list-coloring instance, $B\subset V$ a permissive block with edge boundary $\delta B=\{e_1,e_2,\ldots,e_m\}$ where $e_i=u_iv_i$ for each $i=1,2,\ldots,m$,
and $\pi,\rho\in L^*(B)$ any two proper colorings of $B$.
For every $1\le i\le m$,
\begin{itemize}
\item if a vertex $u\not\in B$ is permissive in $(G,\clist{L})$, then it is permissive in the new instance $(G_B,\clist{L}_{i,\pi,\rho})$;
\item the new instance $(G_B,\clist{L}_{i,\pi,\rho})$ is feasible.
\end{itemize}
For any feasible coloring $\sigma\in L(\Delta)$ of a vertex set $\Delta\subset V$ with $\dist(B,\Delta)\ge 2$, we have
\begin{align*}
\frac{\pr_{G,\clist{L}}(c(B)=\pi\mid \sigma)}{\pr_{G,\clist{L}}(c(B)=\rho\mid \sigma)}
&=
\prod_{i=1}^m
\frac{1-\pr_{G_B,\clist{L}_{i,\pi,\rho}}(c(v_i)=\pi_i\mid \sigma)}{1-\pr_{G_B,\clist{L}_{i,\pi,\rho}}(c(v_i)= \rho_i\mid \sigma)}.
\end{align*}
\end{lemma}

\newcommand{\ProofPinningInvariant}{
When modifying $(G,\clist{L})$ to $(G_B,\clist{L}_{i,\pi,\rho})$, for any vertex $u\not\in B$, every time a color is removed from $L(u)$, at least one of the neighbors of $u$ is also deleted, so $|L(u)|-d(u)$ never decreases, which means that any vertex is permissive in $(G_B,\clist{L}_{i,\pi,\rho})$ if it is permissive in $(G,\clist{L})$. 

We next show that $(G_B,\clist{L}_{i,\pi,\rho})$ is feasible.
Let $R=V\setminus(B\cup\partial B)$ and $\eta$ a proper coloring of subgraph $G[R]$ induced by $R$ (such a coloring must exist or otherwise $(G,\clist{L})$ is not feasible). Recall that every vertex $u\in \partial B$ must remain to have $|L(u)|>d(u)+1$ in $(G_B,\clist{L}_{i,\pi,\rho})$ since $B$ is a permissive block in $(G,\clist{L})$ and $(G_B,\clist{L}_{i,\pi,\rho})$ never reduces the gap $|L(u)|-d(u)$, which means no matter what $\eta$ is, we can always properly color $\partial B$ in a greedy fashion without conflicting with $\eta$, giving us a proper coloring of $(G_B,\clist{L}_{i,\pi,\rho})$.

We then prove the recursion. 
Due to Lemma~\ref{lemma-local-feasible}, both $\pr_{G,\clist{L}}(c(B)=\pi\mid \sigma)$ and $\pr_{G,\clist{L}}(c(B)=\rho\mid \sigma)$ are positive since $\pi$ and $\rho$ are proper on $B$. 
For any $0\le i\le m$, observe that 
\begin{align*}
\pr_{G_B,\clist{L}}(\forall k\le i, c(v_k)\neq \pi_k, \forall k>i, c(v_k)\neq \rho_k\mid \sigma)
&=
\pr_{G_B,\clist{L}_{i,\pi,\rho}}(c(v_i)\neq \pi_i\mid \sigma)\\
&=
1-\pr_{G_B,\clist{L}_{i,\pi,\rho}}(c(v_i)=\pi_i\mid \sigma).
\end{align*}
As argued above we have $|L(v_i)|>d(v_i)+1$ in $(G_B,\clist{L}_{i,\pi,\rho})$ because $v_i$ is a boundary vertex of a permissive block $B$ in $(G,\clist{L})$ and $(G_B,\clist{L}_{i,\pi,\rho})$ never reduces the gap $|L(v_i)|-d(v_i)$.
This implies that the probability $\pr_{G_B,\clist{L}_{i,\pi,\rho}}(c(v_i)=\pi_i\mid \sigma)\le\frac{1}{2}$ because conditioning on any particular coloring of neighbors of $v_i$ there are at least two colors in its list not used by its neighbors. Therefore, we have $\pr_{G_B,\clist{L}}(\forall k\le i, c(v_k)\neq \pi_k, \forall k>i, c(v_k)\neq \rho_k\mid \sigma)>0$, and the following telescopic product is safe to apply:
\begin{align*}
\frac{\pr_{G,\clist{L}}(c(B)=\pi\mid \sigma)}{\pr_{G,\clist{L}}(c(B)=\rho\mid \sigma)}
&=
\frac{\pr_{G_B,\clist{L}}(\forall 1\le i\le m, c(v_i)\neq \pi_i\mid \sigma)}{\pr_{G_B,\clist{L}}(\forall 1\le i\le m, c(v_i)\neq \rho_i\mid \sigma)}\\
&=
\prod_{i=1}^m
\frac{\pr_{G_B,\clist{L}}(\forall k\le i, c(v_k)\neq \pi_k, \forall k>i, c(v_k)\neq \rho_k\mid \sigma)}{\pr_{G_B,\clist{L}}(\forall k< i, c(v_k)\neq \pi_k, \forall k\ge i, c(v_k)\neq \rho_k\mid \sigma)}\\
&=
\prod_{i=1}^m
\frac{\pr_{G_B,\clist{L}_{i,\pi,\rho}}(c(v_i)\neq \pi_i\mid \sigma)}{\pr_{G_B,\clist{L}_{i,\pi,\rho}}(c(v_i)\neq \rho_i\mid \sigma)}\\
&=
\prod_{i=1}^m
\frac{1-\pr_{G_B,\clist{L}_{i,\pi,\rho}}(c(v_i)=\pi_i\mid \sigma)}{1-\pr_{G_B,\clist{L}_{i,\pi,\rho}}(c(v_i)= \rho_i\mid \sigma)}.
\end{align*}
}
\ifabs{}{
\begin{proof}
\ProofPinningInvariant
\end{proof}
}

\ifabs{The following marginal bounds are standard and are proved in Appendix~\ref{appendix-SAW}.}
{
The following bounds for marginal probabilities are quite standard.
}
\begin{lemma}\label{lemma-marginal-bounds}
Given a feasible list-coloring instance $(G,\clist{L})$,
if vertex $v$ has $|L(v)|>d(v)+1$ and $v\not\in\Delta$, then for any feasible coloring $\sigma\in L(\Delta)$ and any $x\in L(v)$, we have
\[
\pr_{G,\clist{L}}(c(v)=x\mid \sigma)\le\frac{1}{|L(v)|-d(v)}.
\]
If vertex $v$ is permissive in $(G,\clist{L})$ and $\dist(v,\Delta)\ge 2$, then for any feasible coloring $\sigma\in L(\Delta)$ and any $x\in L(v)$, we have
\[
\pr_{G,\clist{L}}(c(v)=x\mid \sigma)\ge \frac{1}{|L(v)|2^{d(v)}}.
\]
\end{lemma}
\newcommand{\ProofMarginBounds}{
For the first inequality, conditioning on any coloring of neighbors of $v$, there are at least $|L(v)|-d(v)$ colors in $L(v)$ not used by its neighbors, thus $\pr_{G,\clist{L}}(c(v)=x)\ge\frac{1}{|L(v)|-d(v)}$.

For the second inequality, note that for a permissive $v$ with $\dist(v,\Delta)\ge 2$, $B=\{v\}$ is a permissive block with $\dist(B,\Delta)\ge 2$. Applying the recursion in Lemma~\ref{lemma-partial-pinning-invariant}, we have
\begin{align*}
\frac{\pr_{G,\clist{L}}(c(v)=y\mid \sigma)}{\pr_{G,\clist{L}}(c(v)=x\mid \sigma)}
&=
\prod_{i=1}^{d(v)}
\frac{1-\pr_{G_B,\clist{L}_{i,y,x}}(c(v_i)=y\mid \sigma)}{1-\pr_{G_v,\clist{L}_{i,y,x}}(c(v_i)= x\mid \sigma)},
\end{align*}
for any $x,y\in L(v)$, where $v_1,v_2,\ldots, v_{d(v)}$ are the neighbors of $v$, and each $v_i$ remains to have $|L(v_i)|>d(v_i)+1$ in each new list-coloring instance $(G_B,\clist{L}_{i,y,x})$. Also by Lemma~\ref{lemma-partial-pinning-invariant}, all new instances $(G_B,\clist{L}_{i,y,x})$ are feasible. Thus by the first inequality, we have $\pr_{G_v,\clist{L}_{i,y,x}}(c(v_i)= x\mid \sigma)\le \frac{1}{|L'(v_i)|-d'(v_i)}\le \frac{1}{2}$. Therefore, it holds that
\[
\frac{\pr_{G,\clist{L}}(c(v)=y\mid \sigma)}{\pr_{G,\clist{L}}(c(v)=x\mid \sigma)}\le \prod_{i=1}^{d(v)}\frac{1}{1-\frac{1}{2}}\le 2^{d(v)}.
\]
Summing this over all $y\in L(v)$, we have 
\[
\frac{1}{\pr_{G,\clist{L}}(c(v)=x\mid \sigma)}
=
\sum_{y\in L(v)}\frac{\pr_{G,\clist{L}}(c(v)=y\mid \sigma)}{\pr_{G,\clist{L}}(c(v)=x\mid \sigma)}
\le
|L(v)|2^{d(v)},
\]
which implies $\pr_{G,\clist{L}}(c(v)=x\mid \sigma)\ge \frac{1}{|L(v)|2^{d(v)}}$.
}
\ifabs{}
{
\begin{proof}
\ProofMarginBounds
\end{proof}
}

\ifabs{
The recursion in Lemma~\ref{lemma-partial-pinning-invariant} can imply the following bound for the block-wise  decay of error function $\err(\mu_v^\sigma,\mu_v^\tau)$. The proof generalizes the analysis of the point-wise decay in degree-bounded graphs in~\cite{gamarnik2012strong}, and is put to Appendix~\ref{appendix-SAW}.
}
{
The recursion in Lemma~\ref{lemma-partial-pinning-invariant} can imply the following bound for the block-wise  decay of error function $\err(\mu_v^\sigma,\mu_v^\tau)$, which generalizes the analysis in~\cite{gamarnik2012strong} of the point-wise decay in degree-bounded graphs.
}
\begin{lemma}\label{lemma-error-block-recursion}
Let $(G,\clist{L})$ be a feasible list-coloring instance, $v\in V$ and $B\subset V$ a permissive block around $v$ with edge boundary $\delta B=\{e_1,e_2,\ldots,e_m\}$ where $e_i=u_iv_i$ for each $i=1,2,\ldots,m$. Let $\Delta\subset V$ be a vertex set with $\dist(B,\Delta)\ge 2$, and $\sigma,\tau\in L(\Delta)$ any two feasible colorings of $\Delta$.
Assume $\pi,\rho\in L^*(B)$ to be two proper colorings of $B$ achieving the maximum in the error function:
\[
\err(\mu_{B}^\sigma,\mu_{B}^\tau)=\max_{\pi,\rho\in L^*(B)}\left(\log\left(\frac{\mu_{B}^\sigma(\pi)}{\mu_{B}^\tau(\pi)}\right)-\log\left(\frac{\mu_{B}^\sigma(\rho)}{\mu_{B}^\tau(\rho)}\right)\right).
\]
It holds that
\[
\err(\mu_v^\sigma,\mu_v^\tau)
\le
\sum_{i=1}^m\frac{1}{|L(v_i)|-d(v_i)-1}\cdot\err(\mu_i^\sigma,\mu_i^\tau),
\]
where $\mu_i^\sigma=\mu_{G_B,\clist{L}_{i,\pi,\rho},v_i}^\sigma$ and $\mu_i^\tau=\mu_{G_B,\clist{L}_{i,\pi,\rho},v_i}^\tau$ are the respective marginal distributions of coloring of vertex $v_i$ conditioning on $\sigma$ and $\tau$ in the new list-coloring instance $(G_B,\clist{L}_{i,\pi,\rho})$.
\end{lemma}

\newcommand{\ProofErrorBlockRecursion}{
Let $x,y\in L(v)$ denote the two colors of $v$ achieving the maximum in 
\[
\err(\mu_v^\sigma,\mu_v^\tau)=\max_{x,y\in L(v)}\left(\log\left(\frac{\mu_v^\sigma(x)}{\mu_v^\tau(x)}\right)-\log\left(\frac{\mu_v^\sigma(y)}{\mu_v^\tau(y)}\right)\right).
\]
We then have
\ifabs{
\begin{align*}
\err(\mu_{v}^\sigma,\mu_{v}^\tau)
=&
\log\left(\frac{\pr_{G,\clist{L}}(c(v)=x\mid \sigma)}{\pr_{G,\clist{L}}(c(v)=x\mid \tau)}\right)-\log\left(\frac{\pr_{G,\clist{L}}(c(v)=y\mid \sigma)}{\pr_{G,\clist{L}}(c(v)=y\mid \tau)}\right)\\
=&
\log\left(\dfrac{\sum_{\pi: \pi(v)=x}\pr_{G,\clist{L}}(c(B)=\pi\mid \sigma)}{\sum_{\pi: \pi(v)=x}\pr_{G,\clist{L}}(c(B)=\pi\mid \tau)}\right)\\
&-\log\left(\dfrac{\sum_{\rho: \rho(v)=y}\pr_{G,\clist{L}}(c(B)=\rho\mid \sigma)}{\sum_{\rho: \rho(v)=y}\pr_{G,\clist{L}}(c(B)=\rho\mid \tau)}\right).
\end{align*}
}{
\begin{align*}
\err(\mu_{v}^\sigma,\mu_{v}^\tau)
&=
\log\left(\frac{\pr_{G,\clist{L}}(c(v)=x\mid \sigma)}{\pr_{G,\clist{L}}(c(v)=x\mid \tau)}\right)-\log\left(\frac{\pr_{G,\clist{L}}(c(v)=y\mid \sigma)}{\pr_{G,\clist{L}}(c(v)=y\mid \tau)}\right)\\
&=
\log\left(\dfrac{\sum_{\pi: \pi(v)=x}\pr_{G,\clist{L}}(c(B)=\pi\mid \sigma)}{\sum_{\pi: \pi(v)=x}\pr_{G,\clist{L}}(c(B)=\pi\mid \tau)}\right)-\log\left(\dfrac{\sum_{\rho: \rho(v)=y}\pr_{G,\clist{L}}(c(B)=\rho\mid \sigma)}{\sum_{\rho: \rho(v)=y}\pr_{G,\clist{L}}(c(B)=\rho\mid \tau)}\right).
\end{align*}
}
Due to Lemma~\ref{lemma-local-feasible}, since $B$ is a permissive block and $\mathrm{dist}(B,\Delta)\ge2$,
we have that $\pr_{G,\clist{L}}(c(B)=\pi\mid \sigma)>0$ if and only if $\pi$ is proper on $B$ (and the same also holds for  condition $\tau$). Recall that we use $L^*(B)$ to denote the set of proper colorings of $B$. Therefore, we have 
\begin{align}
\err(\mu_{v}^\sigma,\mu_{v}^\tau)
&\le
\log\max_{\pi\in L^*(B)}\left(\frac{\pr_{G,\clist{L}}(c(B)=\pi\mid \sigma)}{\pr_{G,\clist{L}}(c(B)=\pi\mid \tau)}\right)-\log\min_{\rho\in L^*(B)}\left(\frac{\pr_{G,\clist{L}}(c(B)=\rho\mid \sigma)}{\pr_{G,\clist{L}}(c(B)=\rho\mid \tau)}\right)\notag\\
&=
\max_{\pi,\rho\in L^*(B)}\left(\log\left(\frac{\mu_{B}^\sigma(\pi)}{\mu_{B}^\tau(\pi)}\right)-\log\left(\frac{\mu_{B}^\sigma(\rho)}{\mu_{B}^\tau(\rho)}\right)\right)\notag\\
&=
\log\left(\frac{\mu_{B}^\sigma(\pi)}{\mu_{B}^\tau(\pi)}\right)-\log\left(\frac{\mu_{B}^\sigma(\rho)}{\mu_{B}^\tau(\rho)}\right)= \err(\mu_{B}^\sigma,\mu_{B}^\tau),\label{eq:vertex-margin-to-block-margin}
\end{align}
where $\pi,\rho\in L^*(B)$ are the colorings of $B$ achieving the maximum in
\begin{align*}
\err(\mu_{B}^\sigma,\mu_{B}^\tau)
=
\log\left(\frac{\mu_{B}^\sigma(\pi)}{\mu_{B}^\tau(\pi)}\right)-\log\left(\frac{\mu_{B}^\sigma(\rho)}{\mu_{B}^\tau(\rho)}\right)
=
\log\left(\frac{\mu_{B}^\sigma(\pi)}{\mu_{B}^\sigma(\rho)}\right)-\log\left(\frac{\mu_{B}^\tau(\pi)}{\mu_{B}^\tau(\rho)}\right).
\end{align*}

For $i=1,2,\ldots,m$, we denote by $\mu_{i}^\sigma$ (and $\mu_i^\tau$) the marginal distribution $\mu_{v_i}^\sigma$ (and respective $\mu_{v_i}^\tau$) in the new instance $(G_B,\clist{L}_{i,\pi,\rho})$. Applying the recursion in Lemma~\ref{lemma-partial-pinning-invariant}, we have
\ifabs{
\begin{align*}
\err(\mu_{B}^\sigma,\mu_{B}^\tau)
=&
\log\left(\frac{\mu_{B}^\sigma(\pi)}{\mu_{B}^\sigma(\rho)}\right)-\log\left(\frac{\mu_{B}^\tau(\pi)}{\mu_{B}^\tau(\rho)}\right)\\
=&
\sum_{i=1}^m\left[\log\left(1-\mu_i^\sigma(\pi_i)\right)-\log\left(1-\mu_i^\tau(\pi_i)\right)\right]\\
&-\sum_{i=1}^m\left[\log\left(1-\mu_i^\sigma(\rho_i)\right)-\log\left(1-\mu_i^\tau(\rho_i)\right)\right].
\end{align*}
}
{
\begin{align*}
\err(\mu_{B}^\sigma,\mu_{B}^\tau)
&=
\log\left(\frac{\mu_{B}^\sigma(\pi)}{\mu_{B}^\sigma(\rho)}\right)-\log\left(\frac{\mu_{B}^\tau(\pi)}{\mu_{B}^\tau(\rho)}\right)\\
&=
\sum_{i=1}^m\left[\log\left(1-\mu_i^\sigma(\pi_i)\right)-\log\left(1-\mu_i^\sigma(\rho_i\right)\right]
-\sum_{i=1}^m\left[\log\left(1-\mu_i^\tau(\pi_i)\right)-\log\left(1-\mu_i^\tau(\rho_i)\right)\right]\\
&=
\sum_{i=1}^m\left[\log\left(1-\mu_i^\sigma(\pi_i)\right)-\log\left(1-\mu_i^\tau(\pi_i)\right)\right]
-\sum_{i=1}^m\left[\log\left(1-\mu_i^\sigma(\rho_i)\right)-\log\left(1-\mu_i^\tau(\rho_i)\right)\right].
\end{align*}
}

Let $f(x)=\log\left(1-\mathrm{e}^x\right)$. By the mean value theorem, there exists $\min(x_1,x_2)\le \xi\le \max(x_1,x_2)$ such that 
\[
f(x_1)-f(x_2)=f'(\xi)(x_1-x_2)=\frac{\xi}{1-\xi}(x_2-x_1).
\]
Letting $x_1=\log\mu_i^{\sigma}(\pi_i)$, $x_2=\log\mu_i^\tau(\pi_i)$ (and respectively $x_1=\log\mu_i^{\sigma}(\rho_i)$, $x_1=\log\mu_i^\tau(\rho_i)$), we have for every $1\le i\le m$ that
\begin{align*}
\sum_{i=1}^m\left[\log\left(1-\mu_i^\sigma(\pi_i)\right)-\log\left(1-\mu_i^\tau(\pi_i)\right)\right]
&=
\sum_{i=1}^m\frac{p_i}{1-p_i}\log\left(\frac{\mu_i^\tau(\pi_i)}{\mu_i^\sigma(\pi_i)}\right),\\
\sum_{i=1}^m\left[\log\left(1-\mu_i^\sigma(\rho_i)\right)-\log\left(1-\mu_i^\tau(\rho_i)\right)\right]
&=
\sum_{i=1}^m\frac{p'_i}{1-p'_i}\log\left(\frac{\mu_i^\tau(\rho_i)}{\mu_i^\sigma(\rho_i)}\right),
\end{align*}
where $p_i\le\max\{\mu_i^\sigma(\pi_i),\mu_i^\tau(\pi_i)\}$, $p'_i\le\max\{\mu_i^\sigma(\rho_i),\mu_i^\tau(\rho_i)\}$. 
Note that $v_i$ must have $|L'(v_i)|>d'(v_i)+1$ where $L'(v_i)$ and $d'(v_i)$ are the respective color list and degree of vertex $v_i$ in the new instance $(G_B,\clist{L}_{i,\pi,\rho})$ because $v_i$ is a boundary vertex of a permissive block $B$ in $(G,\clist{L})$ and $|L'(v_i)|-d'(v_i)\ge |L(u)|-d(u)$ since whenever a color is removed from $L(v_i)$ to form $L'(v_i)$, a neighbor of $v_i$ must also be deleted. 
Thus by Lemma~\ref{lemma-marginal-bounds} we have $p_i\le\frac{1}{|L'(v_i)|-d'(v_i)}\le\frac{1}{|L(v_i)|-d(v_i)}$ and $p_i'\le\frac{1}{|L'(v_i)|-d'(v_i)}\le\frac{1}{|L(v_i)|-d(v_i)}$ 
for all $1\le i\le m$, and hence  for all $1\le i\le m$ it holds that
\begin{align*}
\frac{p_i}{1-p_i}\le\frac{1}{|L(v_i)|-d(v_i)-1}\quad\mbox{ and }\quad \frac{p_i'}{1-p_i'}\le\frac{1}{|L(v_i)|-d(v_i)-1}.
\end{align*}
Since $\mu_i^{\sigma}$ and $\mu_i^{\tau}$ are well-defined marginal distributions at vertex $v_i$, we have $\sum_{x\in L(v_i)}\mu_i^\sigma(x)=\sum_{x\in L(v_i)}\mu_i^\tau(x)=1$. Therefore, with a convention $0/0=1$, it can be verified that
\begin{align*}
\max_{y\in L(v_i)}
\log\left(\dfrac{\mu_i^\tau(y)}{\mu_i^\sigma(y)}\right)
\ge 0
\quad\mbox{ and }\quad
\min_{x\in L(v_i)}
\log\left(\dfrac{\mu_i^\tau(x)}{\mu_i^\sigma(x)}\right)
\le 0.
\end{align*}
Therefore, it holds that
\ifabs{
\begin{align*}
\err(\mu_{B}^\sigma,\mu_{B}^\tau)
\le&
\sum_{i=1}^m\frac{p_i}{1-p_i}\log\left(\frac{\mu_i^\tau(\pi_{i})}{\mu_i^\sigma(\pi_{i})}\right)
-\sum_{i=1}^m\frac{p'_i}{1-p'_i}\log\left(\frac{\mu_i^\tau(\rho_{i})}{\mu_i^\sigma(\rho_{i})}\right)\\
\le&
\sum_{i=1}^m\frac{p_i}{1-p_i}\max_{y\in L(v_i)}\log\left(\dfrac{\mu_i^\tau(y)}{\mu_i^\sigma(y)}\right)
-\sum_{i=1}^m\frac{p'_i}{1-p'_i}\min_{x\in L(v)}\log\left(\dfrac{\mu_i^\tau(x)}{\mu_i^\sigma(x)}\right)\\
\le&
\sum_{i=1}^m\frac{1}{|L(v_i)|-d(v_i)-1}\max_{y\in L(v_i)}\log\left(\dfrac{\mu_i^\tau(y)}{\mu_i^\sigma(y)}\right)\\
&-\sum_{i=1}^m\frac{1}{|L(v_i)|-d(v_i)-1}\min_{x\in L(v)}\log\left(\dfrac{\mu_i^\tau(x)}{\mu_i^\sigma(x)}\right)\\
=&
\sum_{i=1}^m\frac{1}{|L(v_i)|-d(v_i)-1}
\max_{x,y\in L(v_i)}\left[\log\left(\dfrac{\mu_i^\sigma(x)}{\mu_i^\tau(x)}\right) - \log\left(\dfrac{\mu_i^\sigma(y)}{\mu_i^\tau(y)}\right)\right]\\
=&
\sum_{i=1}^m\frac{1}{|L(v_i)|-d(v_i)-1}\cdot\err(\mu_i^\sigma,\mu_i^\tau).
\end{align*}
}
{
\begin{align*}
\err(\mu_{B}^\sigma,\mu_{B}^\tau)
&\le
\sum_{i=1}^m\frac{p_i}{1-p_i}\log\left(\frac{\mu_i^\tau(\pi_{i})}{\mu_i^\sigma(\pi_{i})}\right)
-\sum_{i=1}^m\frac{p'_i}{1-p'_i}\log\left(\frac{\mu_i^\tau(\rho_{i})}{\mu_i^\sigma(\rho_{i})}\right)\\
&\le
\sum_{i=1}^m\frac{p_i}{1-p_i}\max_{y\in L(v_i)}\log\left(\dfrac{\mu_i^\tau(y)}{\mu_i^\sigma(y)}\right)
-\sum_{i=1}^m\frac{p'_i}{1-p'_i}\min_{x\in L(v)}\log\left(\dfrac{\mu_i^\tau(x)}{\mu_i^\sigma(x)}\right)\\
&\le
\sum_{i=1}^m\frac{1}{|L(v_i)|-d(v_i)-1}\max_{y\in L(v_i)}\log\left(\dfrac{\mu_i^\tau(y)}{\mu_i^\sigma(y)}\right)
-\sum_{i=1}^m\frac{1}{|L(v_i)|-d(v_i)-1}\min_{x\in L(v)}\log\left(\dfrac{\mu_i^\tau(x)}{\mu_i^\sigma(x)}\right)\\
&=
\sum_{i=1}^m\frac{1}{|L(v_i)|-d(v_i)-1}
\max_{x,y\in L(v_i)}\left[\log\left(\dfrac{\mu_i^\sigma(x)}{\mu_i^\tau(x)}\right) - \log\left(\dfrac{\mu_i^\sigma(y)}{\mu_i^\tau(y)}\right)\right]\\
&=
\sum_{i=1}^m\frac{1}{|L(v_i)|-d(v_i)-1}\cdot\err(\mu_i^\sigma,\mu_i^\tau).
\end{align*}
}

Combining with~\eqref{eq:vertex-margin-to-block-margin}, we have
\begin{align*}
\err(\mu_{v}^\sigma,\mu_{v}^\tau)
&\le 
\sum_{i=1}^m\frac{1}{|L(v_i)|-d(v_i)-1}\cdot\err(\mu_i^\sigma,\mu_i^\tau).
\end{align*}
Recall that $\mu_i^{\sigma}=\mu^\sigma_{v_i}$ and $\mu_i^{\tau}=\mu^\tau_{v_i}$ are the respective marginal distributions at $v_i$ in the new list-coloring instance $(G_B,\clist{L}_{i,\pi,\rho})$. 
}
\ifabs{
}
{
\begin{proof}
\ProofErrorBlockRecursion
\end{proof}
}

\ifabs{
With the above block-wise decay, we are now ready to prove Lemma~\ref{lemma-SAW-decay}, which implies Theorem~\ref{thm-SAW-decay}. 
}
{
Now we are ready to prove Lemma~\ref{lemma-SAW-decay}, which implies Theorem~\ref{thm-SAW-decay}. 
}
\begin{proof}[Proof of Lemma~\ref{lemma-SAW-decay}]
Given a feasible list-coloring instance $(G,\clist{L})$ and a vertex $v$, let $T=T_\SAW(G,v)$ and $S$ a permissive cutset in $T$ separating $v$ and $\Delta$. 
We consider the following procedure:
\begin{enumerate}
\item 
Let $B$ be the \emph{minimal} permissive block around $v$ with edge boundary $\delta B=\{e_1,e_2,\ldots,e_m\}$, where $e_i=u_iv_i$ for $i=1,2,\ldots,m$ (note that more than one $u_i$ or $v_i$ may refer to the same vertex). 
By Lemma~\ref{lemma-error-block-recursion}, 
we have
\begin{align}
\err(\mu_v^\sigma,\mu_v^\tau)
\le
\sum_{i=1}^m\frac{1}{|L(v_i)|-d(v_i)-1}\cdot\err(\mu_i^\sigma,\mu_i^\tau),\label{eq:error-block-recursion}
\end{align}
where $\mu_i^\sigma=\mu_{G_B,\clist{L}_{i,\pi,\rho},v_i}^\sigma$ and $\mu_i^\tau=\mu_{G_B,\clist{L}_{i,\pi,\rho},v_i}^\tau$ are the respective marginal distributions at $v_i$ in the new list-coloring instance $(G_B,\clist{L}_{i,\pi,\rho})$ for the $\pi,\rho\in L^*(B)$ defined in Lemma~\ref{lemma-error-block-recursion}. 
By Lemma~\ref{lemma-partial-pinning-invariant}, all these new list-coloring instances are feasible. 

\item
We identify each $v_i$ with a distinct self-avoiding walk in $G$ from $v$ to $v_i$ through only vertices in $B$ and approaching $v_i$ via the edge $e_i=u_iv_i$. 
Such self-avoiding walk must exist or otherwise $B$ is not minimal.
If there are more than one such self-avoiding walk for a $v_i$, choose an arbitrary one to identify $v_i$ with.
We use $w_i$ to denote this walk to $v_i$.

Note that along every such self-avoiding walk $w_i$ from $v$ to $v_i$, all vertices $u$ except $v$ and $v_i$ must have $|L(u)|\le d(u)+1$ in $(G,\clist{L})$ or otherwise $B$ is not minimal.
Thus by  Definition~\ref{definition-SAW-decay}, 
in quantity $\err_{T,\clist{L},S}$, along every walk $w_i$ from $v$ to $v_i$, at each intermediate vertex $u\not\in\{v, v_i\}$, only a factor of $\contr(u)=1$ is multiplied in $\err_{T,\clist{L},S}$, so we have
\begin{align}
\err_{T,\clist{L},S}\ge \sum_{i=1}^m\frac{1}{|L(v_i)|-d(v_i)-1}\err_{T_{w_i},\clist{L},S},\label{eq:SAW-decay-IH}
\end{align}
where $T_{w_i}$ denotes the subtree of the SAW tree $T$ rooted by the self-avoiding walk $w_i$.

\item
For each $1\le i\le m$,
if the self-avoiding walk $w_i$ corresponds to a vertex in the permissive cutset $S$ in the SAW tree $T$, then $v_i$ itself must be permissive in $(G,\clist{L})$ and $\dist(v_i,\Delta)\ge 2$, both of which continue to hold in the new instance $(G_B,\clist{L}_{i,\pi,\rho})$.
By Lemma~\ref{lemma-marginal-bounds}, we have $\mu_{i}^\sigma(x), \mu_{i}^\tau(x)\in \left[\frac{1}{q2^{q-2}},\frac{1}{2}\right]$ for any $x\in L(v_i)$, thus 
\begin{align}
\err\left(\mu_{i}^\sigma, \mu_{i}^\tau\right)
\le 
2\left(\ln q+q\ln 2\right)\le 3 q;\label{eq:SAW-decay-basis}
\end{align}
and if otherwise, $w_i$ is not in $S$ in the SAW tree $T$, we repeat from the first step for vertex $v_i$ in the new instance $(G_B,\clist{L}_{i,\pi,\rho})$.
\end{enumerate}
We can then apply an induction to prove that $\err(\mu_v^\sigma,\mu_v^\tau)\le \err_{T,\clist{L},S}$, with~\eqref{eq:SAW-decay-basis} as basis, and~\eqref{eq:error-block-recursion} and \eqref{eq:SAW-decay-IH} as induction step. We only need to clarify that each application of~\eqref{eq:error-block-recursion} creates new instances $(G_B,\clist{L}_{i,\pi,\rho})$, while $\err_{T,\clist{L},S}$ is defined using only the original instance $(G,\clist{L})$. This will not cause any issue because by Lemma~\ref{lemma-partial-pinning-invariant}, every new instance $(G_B,\clist{L}_{i,\pi,\rho})$ created during this procedure must be feasible. Moreover, the operation the new instance $(G_B,\clist{L}_{i,\pi,\rho})$ applying on $(G,\clist{L})$ never makes any vertex less permissive, and never increases the multiplicative factor $\frac{1}{|L(v_i)|-d(v_i)-1}$ in the recursion.
\end{proof}

\section{Strong spatial mixing on random graphs}
\label{sec:graph-ssm}
In this section, we prove Theorem~\ref{theorem-ssm-random-graph}, the strong spatial mixing of $q$-coloring of random graph $G(n,d/n)$ with respect to a fixed vertex.
The theorem is proved by applying Theorem~\ref{thm-SAW-decay} to random graph $G(n,d/n)$.
\ifabs{
The following lemma states the existing with high probability of a good permissive cutset in the self-avoiding walk tree of a random graph $G(n,d/n)$. The proof is in Appendix~\ref{appendix-SSM}.
}
{
We first prove the following lemma which states the existing with high probability of a good permissive cutset in the self-avoiding walk tree of a random graph $G(n,d/n)$.
}

\begin{lemma}\label{lemma-permissive-separator-TSAW}
Let $d>1$, $q\ge\alpha d+\beta$ for $\alpha>2$ and $\beta\ge23$, and $t(n)=\omega(1)$ an arbitrary super-constant function.
Let $v\in V$ be arbitrarily fixed and $G=(V,E)$ a random graph draw from $G(n,d/n)$. 
The following event holds with high probability:
for any $t(n)\le t\le \frac{\ln n}{\ln d}$ and any vertex set $\Delta\subset V$ satisfying $\dist_G(v,\Delta)> 2t$, 
there exists a permissive cutset $S$ in $T=T_\SAW(G,v)$ for $v$ and $\Delta$ such that $t\le \dist_T(v,u)<2t$ for all vertices $u\in S$.
\end{lemma}

\newcommand{\ProofPermissiveSeparatorTSAW}{
It is sufficient to show that with high probability, for any $t(n)\le t\le \frac{\ln n}{\ln d}$, and for any path $v=v_0\to v_1\cdots\to v_{2t}$ in $G$, there exists a $t\le i<2 t$ such that $v_i$ is permissive in $(G,[q])$. Without loss of generality, we can assume $q=2d+23$ because the desirable permissiveness will be implied for any greater $q$. 

Let $t(n)\le t\le \frac{\ln n}{\ln d}$.
Consider an arbitrarily fixed tuple $P=(v,v_1,v_2,\ldots,v_{2t})$ of $2t+1$ distinct vertices. We are going to show that there exists a constant $\gamma<1$ such that
\begin{align}
\Pr\left[\forall t\le i<2t, v_i\text{ is not permissive}\mid P\text{ is a path in }G\right]=O\left(\gamma^t d^{-2t}\right). \label{eq:2-permissive-separator}
\end{align}
Note that this implies the statement of the lemma. To see so, assuming~\eqref{eq:2-permissive-separator}, by union bound, the probability that there exists a $P=(v,v_1,v_2,\ldots,v_{2t})$ such that $P$ is a path in $G$ and no $v_i\in P$ with $t\le i<2t$ is permissive is bounded by
\begin{align*}
&\sum_{P=(v,v_1,\ldots,v_{2t})}\Pr\left[\forall t\le i<2t, v_i\in P\text{ is not permissive}\wedge P\text{ is a path }\right]\\
\le
&n^{2t}\left(\frac{d}{n}\right)^{2t} \Pr\left[\forall t\le i<2t, v_i\in P\text{ is not permissive}\mid P\text{ is a path }\right]\\
=
&O\left(\gamma^t\right),
\end{align*}
and by union bound $\sum_{t\ge t(n)}O(\gamma^t)=O(\gamma^{t(n)})=o(1)$, thus with high probability for all $t(n)\le t\le \frac{\ln n}{\ln d}$ and every path $v\to v_1\cdots\to v_{2t}$, we have a $t\le i<2 t$ such that $v_i$ is permissive, which implies the lemma.

Next we prove~\eqref{eq:2-permissive-separator}. Note that if being permissive for each vertex $v_i$ along the path is independent, then the upper bound is immediate, however, they are not. We resolve this by decomposing $G$ into a sequence of subgraphs.

Suppose $P=(v,v_1,v_2,\ldots,v_{2t})$, let $G_P$ be the random graph distributed according to $G(n,d/n)$ conditioning on that $P$ is a path. 
Let $d'(u)$ denote the degree of $u$ in $G$ contributed by the edges not in $P$. We override the definition of permissiveness so that $u$ is permissive if $q>d'(w)+3$ for all neighbors $w$ of $u$ and $w=u$. 
Clearly a vertex is permissive in the original sense if it is permissive in this new definition.
Let $N(u)=\{u\}\cup\{w\mid uw\in G_P\}$ denote the neighborhood of $u$ including $u$ itself in $G_P$.

Let $\ell=\floor{(t-1)/3}$ and $G_0,G_1,\ldots,G_\ell$ a sequence of random graphs defined as follows: 
Let $G_0=G_P$, $V_0=V$ and $N_0=\bigcup_{j=1}^t N(v_j)$; and for $i=1,2,\ldots,\ell$,
\begin{itemize}
\item let $V_{i}=V_{i-1}\setminus N_{i-1}$ and $G_i$ the subgraph of $G_P$ induced by $V_i$;
\item let $N_i=N(v_{t+3i})\cap V_i$ be the neighborhood of $v_{t+3i}$ in $G_i$. 
\end{itemize}
In fact, each $G_i$ is a random graph with vertex set $V_i$ distributed according to $G(|V_i|,d/n)$ conditioning on $(v_{t+3i-1},v_{t+3i},\ldots,v_{2t})$ being a path; and
each $N_i$ can be constructed by sampling each vertex in $V_i\setminus\{v_{t+3i-1},v_{t+3i},v_{t+3i+1}\}$ independent with probability $\frac{d}{n}$ and adjoining with $\{v_{t+3i-1},v_{t+3i+1}\}$ afterwards. 

Let $S_i=\bigcup_{j\le i}N_j$. 
Let $A_i$ denote the event that $|S_i|<(\log n)^4$, and $A=\bigcap_{i=1}^{\ell}A_i$ denote the joint event that all $A_i$ occur simultaneously.
Let $B_i$ denote the event that every vertex in $N(v_{t+3i})$ is adjacent to at most 3 vertices in $S_{i-1}$ by random edges not in $P$,
and $B=\bigcap_{i=1}^{\ell}B_i$ denote the joint event that all $B_i$ occur simultaneously.
Let $n$ be sufficiently large. 
It is easy to see that $A$ holds with probability $1-o(n^{-4})$, because $\sum_{j\le i}|N_j|\ge |S_{i}|>(\log n)^4$ would imply that the max degree of $G$ is $\omega((\log n)^2)$ which holds with probability far less than $o(n^{-4})$.
It is also east to see $B$ holds with probability $1-o(n^{-3})$. To see so, conditioning on event $A$, which holds with probability $1=o(n^{-4})$, for every $i$, $|N(v_{t+3i})|<(\log n)^4$ with probability $1-o(n^{-4})$, and the probability that there is a vertex in $N(v_{t+3i})$ being adjacent to 4 vertices in $S_i$ by random edges (a Bernoulli trial occurs with probability $d/n$) is $O(n^{-4}\mathrm{polylog} n)$. By union bound, with probability $1-o(n^{-3})$ all $B_i$ hold simultaneously. 

Note that when $B_i$ occurs, for all vertices $u$ in the neighborhood $N(v_{t+3i})$ we have $d'(u)$ reduced by at most 3 if restricted on $G_i$, so the permissiveness of $v_{t+3i}$ in $(G_P,[q])$ is implied by the permissiveness of $v_{t+3i}$ in $(G_i,[q-3])$. Therefore, it holds that
\ifabs{
\begin{align}
&\Pr\left[\,\forall t\le i<2t, v_i\text{ is not permissive in }(G_P,[q])\,\right]\notag\\
\le 
&\Pr\left[\,A\wedge B\wedge \forall 1\le i\le\ell, v_{t+3i}\text{ is not permissive in }(G_P,[q-3])\,\right]+o(n^{-3})\notag\\
\le
&\Pr\left[\,A\wedge\forall 1\le i\le\ell,  v_{t+3i}\text{ is not permissive in }(G_i,[q-3])\,\right]+o(n^{-3})\notag\\
=
&\prod_{i=1}^{\ell}\Pr[\,v_{t+3i}\text{ is not permissive in }G_i\notag\\
&\qquad\quad\mid A_{i-1},  \forall j\le i-1, v_{t+3j}\text{ is not permissive in }G_j\,]+o(1)\notag\\
=
&\prod_{i=1}^{\ell}\Pr\left[\,v_{t+3i}\text{ is not permissive in }(G_i,[q-3])\mid A_{i-1}\,\right]+o(n^{-3}),\label{eq:not-2-permissive}
\end{align}
}{
\begin{align}
&\Pr\left[\,\forall t\le i<2t, v_i\text{ is not permissive in }(G_P,[q])\,\right]\notag\\
\le 
&\Pr\left[\,A\wedge B\wedge \forall 1\le i\le\ell, v_{t+3i}\text{ is not permissive in }(G_P,[q-3])\,\right]+o(n^{-3})\notag\\
\le
&\Pr\left[\,A\wedge\forall 1\le i\le\ell,  v_{t+3i}\text{ is not permissive in }(G_i,[q-3])\,\right]+o(n^{-3})\notag\\
=
&\prod_{i=1}^{\ell}\Pr\left[\,v_{t+3i}\text{ is not permissive in }G_i\mid A_{i-1},  \forall j\le i-1, v_{t+3j}\text{ is not permissive in }G_j\,\right]+o(1)\notag\\
=
&\prod_{i=1}^{\ell}\Pr\left[\,v_{t+3i}\text{ is not permissive in }(G_i,[q-3])\mid A_{i-1}\,\right]+o(n^{-3}),\label{eq:not-2-permissive}
\end{align}
}
where the last equation is guaranteed by that the permissiveness of $v_{t+3i}$ in $G_i$ is independent of that of $v_{t+3j}$ in $G_j$ for any previous $j<i$, even though the sequence $G_i$ itself is not independent.

Recall that graph $G_i$ is distributed according to $G(|V_i|,d/n)$ conditioning on $(v_{t+3i-1},v_{t+3i},\ldots,v_{2t})$ being a path. Conditioning on $A_{i-1}$, we have $|V_i|\ge n-|S_{i-1}|=(1-o(1))n$. We then analyze the probability of a vertex $u$ being not permissive in such a random graph. Let $X_0=d'(u)$ be the degree of $u$ in $G_i$ contributed by edges not in the path $P$ and for the up to $X_0+2$ neighbors $u_1,u_2,\ldots,u_{X_0+2}$ of $u$, let $X_i=d'(u_i)$ for $i=1,2,\ldots,X_0+2$. Clearly each of $X_i$ is a binomial random variable distributed according to $B(n',d/n)$ with $n'=(1-o(1))n$. Vertex $u$ is permissive in $(G_i,[q-3])$ if $q-3>X_i+3$ for all $i$, thus if $u$ is not permissive in $(G_i,[q-3])$ then either $X_0\ge q-6=2d+17$ or $X_0<2d+17$ and one of the at most $X_0+2\le 2d+18$ random variables $X_1,X_2,\ldots,X_{2d+18}$ has $X_i\ge q-6=2d+17$. By union bound and Chernoff bound, this probability is upper bounded by
\ifabs{
\begin{align*}
&\Pr[\,u\text{ is not permissive in }(G_i,[q-3])\mid A_{i-1}\,]\\
\le&
(2d+19)\Pr[X_0\ge 2d+17]\\
\le& 
(2d+19)\left(\frac{\mathrm{e}^{d+17}}{(2+17/d)^{2d+17}}\right)^{(1-o(1))}\\
<&
\left(\frac{3}{4d^6}\right),
\end{align*}
}
{
\begin{align*}
\Pr[\,u\text{ is not permissive in }(G_i,[q-3])\mid A_{i-1}\,]
\le&
(2d+19)\Pr[X_0\ge 2d+17]\\
\le& 
(2d+19)\left(\frac{\mathrm{e}^{d+17}}{(2+17/d)^{2d+17}}\right)^{(1-o(1))}\\
<&
\left(\frac{3}{4d^6}\right),
\end{align*}
}
for any $d$ and all sufficiently large $n$.

Therefore, \eqref{eq:not-2-permissive} can be bounded as 
\begin{align*}
\Pr\left[\,\forall t\le i<2t, v_i\text{ is not permissive in }(G_P,[q])\,\right]
&\le 
\left(\frac{3}{4d^6}\right)^\ell+o(n^{-3})\\
&=O((3/4)^{t/3}d^{-2t}),
\end{align*}
for $\ell=\floor{(t-1)/3}$ and for any $t\le \frac{\ln n}{\ln d}$.
This proves \eqref{eq:2-permissive-separator}, which implies the lemma.
}
\ifabs{}
{
\begin{proof} 
\ProofPermissiveSeparatorTSAW
\end{proof}
}

\ifabs{
We then observe that the quantity $\err_{T,\clist{L},S}$ decays fast on average. The proof is also in Appendix~\ref{appendix-SSM}.
}
{We then observe that the quantity $\err_{T,\clist{L},S}$ defined in Definition~\ref{definition-SAW-decay} decays fast on average.
}
\begin{lemma}\label{lemma-decay-expectation}
Let $f_q(x)$ be a piecewise function defined as 
\[
f_q(x)=\begin{cases}
\frac{1}{q-x-1} & \mbox{if }x\le q-2,\\
1 & \mbox{otherwise.}
\end{cases}
\]
Let $X$ be a random variable distributed according to binomial distribution $B(n,\frac{d}{n})$ where $d=o(n)$. For $q\ge 2 d+4$, it holds that $\E{f_q(X)}<\frac{1}{d}$.
\end{lemma}

\newcommand{\ProofDecayExpectation}{
We denote that $p(k)=\binom{n}{k} \left(\frac{d}{n}\right)^k \left(1-\frac{d}{n}\right)^{n-k}$.
Note that 
\[
1-\E{f_q(X)}=\sum_{k=1}^{q-2}\left(1-\frac{1}{q-k-1}\right)p(k)
\]
is nondecreasing in $q$, so it is sufficient to prove the inequality $1-\E{f_q(X)}>1-\frac{1}{d}$ when $q=\lceil 2d+4\rceil$, which is the following inequality
\begin{equation*}
\sum_{k=1}^{\lceil 2d+2\rceil}g(k)p(k)>1-\frac{1}{d},
\end{equation*}
where the function $g(x)$ is defined as $g(x)=1-\frac{1}{2d+3-x}$.

The function $g(x)$ can be approximated by the following polynomial form: 
\begin{align*}
\tilde{g}(x)&=\frac{d+2}{d+3}-\frac{(x-d)}{(d+3)^2}-\frac{(x-d)^2}{(d+3)^3}-\frac{(x-d)^3}{(d+3)^4}-\frac{(x-d)^4}{(d+3)^5}-\frac{(x-d)^5}{(d+3)^6}-\frac{(x-d)^6}{(d+3)^6}.
\end{align*}
As an illustration, Figure~\ref{fig:poly-approx} shows how well polynomial function $\tilde{g}(x)$ approximates $g(x)$.
\begin{figure}
\centering
  \includegraphics[width=3.4in]{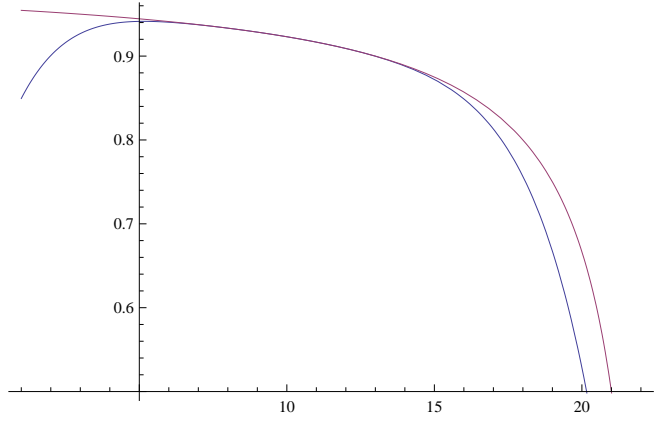}\\
  \caption{$g(x)$ and $\tilde{g}(x)$ when $d=10$.}\label{fig:poly-approx}
\end{figure}
Indeed, it can be verified that
\begin{equation*}
g(x)-\tilde{g}(x)=\frac{(x-d)^6(2d+2-x)}{(d+3)^6(2d+3-x)},
\end{equation*}
thus we have $g(x)>\tilde{g}(x)$ for $x<2d+2$ and $g(2d+2)=\tilde{g}(2d+2)$, and hence it holds that
\begin{equation*}
\sum_{k=1}^{\lceil 2d+2\rceil}g(k)p(k)
>
\sum_{k=1}^{\lfloor 2d+2\rfloor}\tilde{g}(k)p(k).
\end{equation*}
We then show that $\sum_{k=1}^{\lfloor 2d+2\rfloor}\tilde{g}(k)p(k)>1-\frac{1}{d}$, which proves the lemma.

Since $\tilde{g}(x)$ is just a polynomial of $x$ with degree 6, its expectation with binomial input $X\sim B(n,\frac{d}{n})$ can be calculated as
\begin{align*}
\E{\tilde{g}(X)}
=
\frac{1}{(d+3)^6}\Big(
&d^6+17d^5+119d^4+422d^3+867d^2+1012d+486 \\
&+\frac{1}{n}(d^5+63d^4+290d^3+121d^2)-\frac{1}{n^2} (50d^5+498d^4+284d^3)\\
&+\frac{1}{n^3}(15d^6+416d^5+468d^4)-\frac{1}{n^4}(130d^6+384d^5)+\frac{120d^6}{n^5}\Big).
\end{align*}
When $d=o(n)$, we have
$\frac{1}{n}(d^5+63d^4+290d^3+121d^2)-\frac{1}{n^2} (50d^5+498d^4+284d^3)+\frac{1}{n^3}(15d^6+416d^5+468d^4)-\frac{1}{n^4}(130d^6+384d^5)+\frac{120d^6}{n^5}
\ge0$,
thus $\E{\tilde{g}(X)}$ can be lower bounded as
\begin{align*}
\E{\tilde{g}(X)}
&\ge
\frac{d^6+17d^5+119d^4+422d^3+867d^2+1012d+486}{(d+3)^6}\\
&=
\left(1-\frac{1}{d}\right)
+\frac{2d^5+17d^4+192d^3+769d^2+1215d+729}{d(d+3)^6}.
\end{align*}
On the other hand, $\E{\tilde{g}(X)}$ can be decomposed as
\begin{align*}
\E{\tilde{g}(X)}
&=
\tilde{g}(0)\left(1-\frac{d}{n}\right)^n
+\sum_{k=1}^{\lfloor 2d+2\rfloor}\tilde{g}(k)p(k)
+\sum_{k=\lfloor 2d+3\rfloor}^n\tilde{g}(k)p(k).
\end{align*}
For $x=0$, we have $\tilde{g}(0)=1-\frac{(d+3)^6+d^6(2d+2)}{(d+3)^6(2d+3)}<1$.
For $x=\floor{2d+3}$,  we have $\tilde{g}(\floor{2d+3})=g(\floor{2d+3})-\left(\frac{\floor{2d+3}-d}{d+3}\right)^6\frac{2d-\floor{2d}-1}{2d-\floor{2d}}\le 1-\frac{1}{2d-\floor{2d}}-\frac{2d-\floor{2d}-1}{2d-\floor{2d}}\le 0$.
And when $x\ge  2d+3 $, $\tilde{g}(x)$ is monotonically decreasing in $x$, thus $\tilde{g}(x)\le \tilde{g}(2d+3)<0$ for all $x\ge 2d+3$. 
Therefore, it holds that
\begin{align*}
\sum_{k=1}^{\floor{2d+2} }\tilde{g}(k)p(k)-\left(1-\frac{1}{d}\right)
&>
\E{\tilde{g}(X)}-\left(1-\frac{1}{d}\right)
-\left(1-\frac{d}{n}\right)^n\\
&>
\frac{2d^5+17d^4+192d^3+769d^2+1215d+729}{d(d+3)^6}
-\mathrm{e}^{-d},
\end{align*}
which is greater than $0$, since the function
\[
\frac{2d^5+17d^4+192d^3+769d^2+1215d+729}{d(d+3)^6}\mathrm{e}^d
\]
is unimodal and has minimum >1. Therefore, we have
\[
\sum_{k=1}^{\floor{2d+2} }\tilde{g}(k)p(k)>1-\frac{1}{d},
\]
which proves the lemma as argued in the beginning of the proof.
}
\ifabs{}{
\begin{proof}
\ProofDecayExpectation
\end{proof}
}

We then prove a strong spatial mixing theorem with the norm of error function $\err(\mu_v^\sigma,\mu_v^\tau)$.
\begin{lemma}\label{lemma-error-decay-random-graph}
Let $d>1$, $q\ge\alpha d+\beta$ for $\alpha>2$ and $\beta\ge23$, and $t(n)=\omega(1)$ an arbitrary super-constant function.
Let $v\in V$ be arbitrarily fixed and $G=(V,E)$ a random graph draw from $G(n,d/n)$. 
There exist constants $C_1,C_2>0$ depending only on $d$ and $q$ such that with high probability $G$ is $q$-colorable and
\[
\err(\mu_v^\sigma,\mu_v^\tau)\le C_1  \exp(-C_2 \dist(v,\Delta))
\]
for any feasible $q$-colorings $\sigma,\tau\in[q]^\Lambda$ partially specified on a subset $\Lambda\subset V$ of vertices, such that $\sigma$ and $\tau$ differ only on a subset $\Delta\subseteq\Lambda$ with $\dist(v,\Delta)\ge t(n)$.
\end{lemma}
\newcommand{\ProofErrorDecayRandomGraph}{
Fix $v\in V$. Let $T=T_\SAW(G,v)$ be the self-avoiding walk tree of $G$. 
Note that for any set $S$ of vertices in $T$, the quantity $\err_{T,[q],S}$ in Definition~\ref{definition-SAW-decay} is always well-defined (even when $G$ is not $q$-colorable) and is a sum of products of the form $3q\prod_{u\in P\atop u\neq v}\contr(u)$, where each product is taken along a self-avoiding walk $P=(v,v_1,\ldots,v_k)$ from $v$ to a vertex $v_k\in S$, and the contribution $\contr(u)$ of each vertex $u$ in the product is given by the following piecewise function 
\[
\delta(u)=\begin{cases}
\frac{1}{q-d_G(u)-1} & \text{if }q>d_G(u)+1,\\
1 & \text{otherwise.}
\end{cases}
\]
Therefore, for any set $S$ of vertices in $T$ satisfying $t\le \dist_T(v,u)<2t$ for all $u\in S$, we have
\begin{align*}
\err_{T,[q],S}
&\le 
\sum_{k=t}^{2t-1}\sum_{P=(v,v_1,\ldots,v_k)}\mathrm{I}[P\text{ is a path]}\cdot 3q\prod_{i=1}^k\contr(v_i),
\end{align*}
where $\mathrm{I}[P\text{ is a path}]$ is the indicator random variable for the event that $P$ is a path in the random graph $G$.

Fix an arbitrary $t(n)\le t\le\frac{\ln n}{\ln d}$. Consider $\err_t=\max_{S}\err_{T,[q],S}$ where the maximum is taken over all vertex set $S$ in $T$ satisfying $t\le \dist_T(v,u)<2t$ for all $u\in S$. By the above argument and linearity of expectation, we have
\begin{align}
\E{\err_t}
&\le
\sum_{k=t}^{2t-1} n^k\left(\frac{d}{n}\right)^k \cdot\E{3q\prod_{i=1}^k\contr(v_i)\biggm{|} P=(v,v_1,\ldots,v_k)\text{ is a path}}\notag\\
&=
3q\sum_{k=t}^{2t-1}d^k\cdot \E{\prod_{i=1}^k\contr(v_i)\biggm{|} P=(v,v_1,\ldots,v_k)\text{ is a path}}.\label{eq:lemma-decay-random-graph}
\end{align}
We then calculate the expectations.
Fix a tuple $P=(v,v_1,\ldots,v_k)$. The degrees $d(v_i)$ in random graph $G$ are not independent. We then construct an independent sequence whose product dominates the $\prod_{i=1}^k\contr(v_i)$ as follows.

Conditioning on $P=(v,v_1,\ldots,v_k)$ being a path in $G$. 
Let $X_1,X_2,\ldots, X_k$ be random variables such that each $X_i$ represents the number of edges between $v_i$ and vertices in $V\setminus\{v_1,\ldots,v_k\}$; and let $Y$ be a random variable representing the number of edges between vertices in $\{v_1,\ldots,v_k\}$ except for the edges in the path $P=(v,v_1,\ldots,v_k)$.
Then $X_1,X_2,\ldots, X_k, Y$ are mutually independent binomial random variables with each $X_i$ distributed according to $B(n-k,\frac{d}{n})$ and $Y$ distributed according to $B({k\choose 2}-k+1,\frac{d}{n})$, and for each $v_i$ in the path we have $d(v_i)=X_i+2+Y_i$ with some $Y_1+Y_2+\cdots +Y_k=2Y$.

Note that $\contr(v_i)=f_q(d(v_i))$ where function $f_q(x)$ is as defined in Lemma~\ref{lemma-decay-expectation}. Note that the ratio $f_q(x)/f_q(x-1)$ is always upper bounded by 2, and we have the identity $f_q(x+1)=f_{q-1}(x)$. 
Thus, conditioning on that $P=(v,v_1,\ldots,v_k)$ is a path, the product $\prod_{i=1}^k\contr(v_i)$ can be bounded as follows:  
\begin{align*}
\prod_{i=1}^k\contr(v_i)
&=
\prod_{i=1}^kf_q(X_i+2+Y_i)
\le 
2^{2Y}\prod_{i=1}^kf_q(X_i+2)
=
4^Y\prod_{i=1}^kf_{q-2}(X_i).
\end{align*}
Let $d'=(q-6)/2$, thus we have $d'>d$. Let $X$ be binomial random variable distributed according to $B(n,\frac{d'}{n})$, thus $X$ probabilistically dominates every $X_i$ whose distribution is $B(n-k,\frac{d}{n})$. Since $X_1,X_2,\ldots,X_k,Y$ are mutually independent conditioning on $P=(v,v_1,\ldots,v_k)$ being a path in $G$, for any $P=(v,v_1,\ldots,v_k)$ we have
\begin{align*}
\E{\prod_{i=1}^k\contr(v_i)\biggm{|} P\text{ is a path}}
&\le
\E{4^Y\prod_{i=1}^kf_{q-2}(X_i)}
\le
\E{4^Y}\E{f_{q-2}(X)}^k.
\end{align*}
Recall that $Y\sim B\left({k\choose 2}-k+1,\frac{d}{n}\right)$, the expectation $\E{4^Y}$ can be bounded as
\begin{align*}
\E{4^Y}
&\le
\sum_{\ell=0}^{k^2} 4^\ell {k^2 \choose \ell}\left(\frac{d}{n}\right)^\ell\left(1-\frac{d}{n}\right)^{k^2-\ell}
=
\left(1+\frac{3d}{n}\right)^{k^2}
\le \exp\left(\frac{3dk^2}{n}\right).
\end{align*}
Recall that $d'=(q-6)/2$ and $X\sim B(n,\frac{d'}{n})$. We have $q-2= 2d'+4$. Due to Lemma~\ref{lemma-decay-expectation}, we have
$\E{f_{q-2}(X)}<\frac{1}{d'}$. Since we assume $q\ge \alpha d+23$, we have $d'=(q-6)/2>\frac{\alpha}{2}d$ for an $\alpha>2$.
Therefore, we have
\begin{align*}
\E{\prod_{i=1}^k\contr(v_i)\biggm{|} P\text{ is a path}}
&\le
\exp\left(\frac{3dk^2}{n}\right)\left(\frac{1}{d'}\right)^k
\le
\frac{1}{d^k}\exp\left(-k\ln\frac{\alpha}{2}+\frac{3dk^2}{n}\right).
\end{align*}
Combined with~\eqref{eq:lemma-decay-random-graph}, we have
\begin{align*}
\E{\err_t}
&\le
3q\sum_{k=t}^{2t-1}\exp\left(-k\ln\frac{\alpha}{2}+\frac{3dk^2}{n}\right)
\le
6qt\exp\left(-t\ln\frac{\alpha}{2}+o(1)\right).
\end{align*}
By Markov's inequality $\err_t\ge 6qt\exp\left(-\frac{t}{2}\ln\frac{\alpha}{2}+o(1)\right)$ with probability at most $\exp\left(-\frac{t}{2}\ln\frac{\alpha}{2}\right)$. By union bound the probability that there exists a $t(n)\le t\le\frac{\ln n}{\ln d}$ such that $\err_t\ge 6qt\exp\left(-\frac{t}{2}\ln\frac{\alpha}{2}+o(1)\right)$ is at most $\sum_{t\ge t(n)}\exp\left(-\frac{t}{2}\ln\frac{\alpha}{2}\right)\le \exp\left(-\frac{t(n)}{2}\ln\frac{\alpha}{2}\right)/(1-\sqrt{2/\alpha})=o(1)$, thus with high probability we have
\[
\err_t\le 6qt\exp\left(-\frac{t}{2}\ln\frac{\alpha}{2}+o(1)\right)
\] 
for all $t(n)\le t\le\frac{\ln n}{\ln d}$.

We define four good events.
\begin{list}{}{}
\item $E_1$ :  $G$  is $q$-colorable;
\item $E_2$ :  $\mathrm{diam}(G)\le \frac{C_d\ln n}{\ln d}$ where $C_d$ is a large constant depending only on $d$;
\item $E_3$ :  the event defined in Lemma~\ref{lemma-permissive-separator-TSAW};
\item $E_4$ :  $\err_t\le 6qt\exp\left(-\frac{t}{2}\ln\frac{\alpha}{2}+o(1)\right)\text{ for all } t(n)\le t\le\frac{\ln n}{\ln d}$.
\end{list}
It is well-known that with our choice of $d$ and $q$, with high probability $G$ is $q$-colorable~\cite{achlioptas2005two} and according to~\cite{chung2001diameter}, with a properly chosen constant $C_d$ depending only on $d$, the diameter of $G$ has $\mathrm{diam}(G)\le \frac{C_d\ln n}{\ln d}$ with high probability, thus $E_1$ and $E_2$ both occur with high probability. By Lemma~\ref{lemma-permissive-separator-TSAW}, $E_3$ occurs with high probability, and we just prove that $E_4$ occurs with high probability. By union bound, with high probability all these four good events occur simultaneously.

With $E_1$ occurring, the Gibbs measure $\mu$ of $q$-coloring of $G$ is well-defined. 
With $E_3$ occurring, for any $t(n)\le t\le\frac{\ln n}{\ln d}$ and any vertex set $\Delta\subset V$ satisfying $\dist_G(v,\Delta)>2t$, there always exists a permissive cutset $S$ in $T=T_\SAW(G,v)$ for $v$ and $\Delta$ such that $t\le \dist_T(v,u)<2t$ for all $u\in S$,
which by Theorem~\ref{thm-SAW-decay}, implies that $\err(\mu_v^\sigma,\mu_v^\tau)\le\err_{T,[q],S}$. 
Since $\err_t$ is the maximum of $\err_{T,[q],S}$ over all such $S$ with $t\le \dist_T(u,v)\le 2t$ for every $u\in S$, we have $\err(\mu_v^\sigma,\mu_v^\tau)\le\err_{T,[q],S}\le \err_t$. 
With $E_4$ occurring, we have $\err_t\le 6qt\exp\left(-\frac{t}{2}\ln\frac{\alpha}{2}+o(1)\right)$. 
With $E_2$ occurring, it always holds that $\dist_G(v,\Delta)\le\frac{10000\ln n}{\ln d}$, thus by setting $t=\min(\dist_G(v,\Delta), \frac{\ln n}{\ln d})$ we can always guarantee both $t=\Theta(\dist_G(v,\Delta))$ and $t(n)\le t\le \frac{\ln n}{\ln d}$.  
Combining everything together, we have
\[
\err(\mu_v^\sigma,\mu_v^\tau)\le\err_{T,[q],S}
\le 
\err_t\le 6qt\exp\left(-\frac{t}{2}\ln\frac{\alpha}{2}+o(1)\right),
\]
with $t=\Theta(\dist_G(v,\Delta))$. When $\alpha >2$, and $n$ sufficiently large, we have 
\[
6qt\exp\left(-\frac{t}{2}\ln\frac{\alpha}{2}+o(1)\right)\le C_1 \exp(-C_2 \dist_G(v,\Delta))
\] 
for some universally fixed constants $C_1,C_2>0$ depending only on $q$. And all these hold together with probability $1-o(1)$.
}
\ifabs{
\begin{proof}[Sketch of Proof]
We only give a sketch of the proof. The full proof is given in Appendix~\ref{appendix-SSM}.

Fix $v\in V$. Let $T=T_\SAW(G,v)$ be the self-avoiding walk tree of $G$. 
Fix an arbitrary $t(n)\le t\le\frac{\ln n}{\ln d}$. Consider $\err_t=\max_{S}\err_{T,[q],S}$ where the maximum is taken over all vertex set $S$ in $T$ satisfying $t\le \dist_T(v,u)<2t$ for all $u\in S$. 
By enumerating all self-avoiding walks $P=(v,v_1,\ldots,v_k)$ from $v$ to a vertex $v_k\in S$, we have
\begin{align*}
\E{\err_t}
&\le
3q\sum_{k=t}^{2t-1}d^k\cdot \E{\prod_{i=1}^k f_q(d_G(v_i))\biggm{|} P=(v,v_1,\ldots,v_k)\text{ is a path}},
\end{align*}
where the function $f_q(x)$ is as defined in Lemma~\ref{lemma-decay-expectation}.
We then calculate the expectations.
Fix a tuple $P=(v,v_1,\ldots,v_k)$. 
We construct an independent sequence whose product dominates the $\prod_{i=1}^k f_q(d_G(v_i))$.

Conditioning on $P=(v,v_1,\ldots,v_k)$ being a path in $G$. 
Let $X_1,X_2,\ldots, X_k$ be such that each $X_i$ is the number of edges between $v_i$ and vertices in $V\setminus\{v_1,\ldots,v_k\}$; and let $Y$ be the number of edges between vertices in $\{v_1,\ldots,v_k\}$ except for the edges in the path $P=(v,v_1,\ldots,v_k)$.
Then $X_1,X_2,\ldots, X_k, Y$ are mutually independent binomial random variables, 
and for each $v_i$ in the path we have $d_G(v_i)=X_i+2+Y_i$ for some $Y_1+Y_2+\cdots +Y_k=2Y$.


Due to the property of function $f_q(x)$ we can bound that
\begin{align*}
\prod_{i=1}^k f_q(d_G(v_i))
&=
\prod_{i=1}^kf_q(X_i+2+Y_i)
\le 
4^Y\prod_{i=1}^kf_{q-2}(X_i).
\end{align*}
Since $X_1,X_2,\ldots,X_k,Y$ are mutually independent conditioning on $P$ is a path,
\begin{align*}
\E{\prod_{i=1}^k f_q(d_G(v_i))\biggm{|} P\text{ is a path}}
&\le
\E{4^Y\prod_{i=1}^kf_{q-2}(X_i)}
\le
\E{4^Y}\E{f_{q-2}(X)}^k,
\end{align*}
where $\E{f_{q-2}(X)}$ can be upper bounded by Lemma~\ref{lemma-decay-expectation}, and $\E{4^Y}$ by the binomial theorem. Then a calculation gives
\begin{align*}
\E{\err_t}
\le
3q\sum_{k=t}^{2t-1}
\E{\prod_{i=1}^k f_{q}(d_G(v_i))\biggm{|} P\text{ is a path}}
&\le
\exp\left(-\Omega(t)\right).
\end{align*}
By Markov's inequality and union bound, with high probability we have
$\err_t\le \exp\left(-\Omega(t)\right)$ 
for all $t(n)\le t\le\frac{\ln n}{\ln d}$.

By~\cite{achlioptas2005two}, w.h.p.~$G$ is $q$-colorable.
By Lemma~\ref{lemma-permissive-separator-TSAW}, w.h.p.~we have good permissive cutset $S$ satisfying the conditions in Lemma~\ref{lemma-permissive-separator-TSAW}, which by Theorem~\ref{thm-SAW-decay}, implies that $\err(\mu_v^\sigma,\mu_v^\tau)\le\err_{T,[q],S}\le \err_t\le \exp\left(-\Omega(t)\right)$.
By~\cite{chung2001diameter}, w.h.p.~the diameter of $G$ is in $O(\frac{\ln n}{\ln d})$, thus we can choose $t=\Theta(\dist(v,\Delta))$ with a $t(n)\le t\le\frac{\ln n}{\ln d}$, which gives us $\err(\mu_v^\sigma,\mu_v^\tau)\le\exp\left(-\Omega(\dist(v,\Delta))\right)$.
\end{proof}
}
{
\begin{proof}
\ProofErrorDecayRandomGraph
\end{proof}
}

\newcommand{\ProofSSMRandomGraph}{
We denote $\Lambda=\partial R$.
Due to Lemma~\ref{lemma-error-decay-random-graph}, with $q\ge \alpha d+23$ for any $\alpha >2$, with high probability the random graph $G\sim G(n,d/n)$ is $q$-colorable and for any feasible $q$-colorings $\sigma,\tau$ partially specified on $\Lambda\subset V$ such that $\sigma,\tau$ differ on a subset $\Delta\subseteq \Lambda$ with $\dist(v,\Delta)\ge t(n)$, it holds that
\begin{align*}
\err(\mu_v^\sigma,\mu_v^\tau)
&=\max_{x,y\in[q]}\left(\log\frac{\mu_v^\sigma(x)}{\mu_v^\tau(x)}-\log\frac{\mu_v^\sigma(y)}{\mu_v^\tau(y)}\right)
\le
\epsilon(\dist(v,\Delta)),
\end{align*}
where $\epsilon(\dist(v,\Delta))=\left(-C_2 \dist(v,\Delta)\right)$ for constants $C_1,C_2>0$ depending only on $d$ and $q$.

Since $\sum_{x}\mu_v^\sigma(x)=\sum_{x}\mu_v^\tau(x)=1$, we have $-\err(\mu_v^\sigma,\mu_v^\tau)\le \log\frac{\mu_v^\sigma(x)}{\mu_v^\tau(x)}\le \err(\mu_v^\sigma,\mu_v^\tau)$.
Let $n$ be sufficiently large so that $\dist(v,\Delta)\ge t(n)$ satisfies the followings:
\[
1-2\epsilon(\dist(v,\Delta))\le \exp\left(\epsilon(\dist(v,\Delta))\right)\le 1+2\epsilon(\dist(v,\Delta)).
\]
For any $x\in[q]$, we have (with the convention 0/0=1) 
\begin{align*}
1-2\epsilon(\dist(v,\Delta))\le
\frac{\Pr[c(v)=x\mid\sigma]}{\Pr[c(v)=x\mid\tau]}
=
\frac{\mu_v^\sigma(x)}{\mu_v^\tau(x)}
\le 1+2\epsilon(\dist(v,\Delta)),
\end{align*}
where $\epsilon(t)=C_1 q\exp\left(-C_2 t\right)$, which implies that 
\[
|\Pr[c(v)=x\mid\sigma]-\Pr[c(v)=x\mid\tau]|\le C_3\exp(-C_4 \dist(v,\Delta)),
\]
for some constants $C_3,C_4>0$ depending only on $d$ and $q$.
}

\ifabs{
With Lemma~\ref{lemma-error-decay-random-graph},  the proof of Theorem~\ref{theorem-ssm-random-graph} is immediate, which is in Appendix~\ref{appendix-SSM}.
}
{
With Lemma~\ref{lemma-error-decay-random-graph},  the proof of Theorem~\ref{theorem-ssm-random-graph} is immediate, which is stated as follows.
\begin{proof}[Proof of Theorem~\ref{theorem-ssm-random-graph}]
\ProofSSMRandomGraph
\end{proof}
}



\ifabs{
\section*{Appendix}
\appendix
\input{appendix}
}
{}

\end{document}